\documentclass[preprint,final]{elsarticle}

\usepackage[T1]{fontenc} % Use 8-bit encoding that has 256 glyphs
\usepackage{microtype} % Slightly tweak font spacing for aesthetics

\usepackage{amssymb, amsmath, amsthm}

\numberwithin{equation}{section}

\usepackage{graphicx}

\usepackage[hmarginratio=1:1,top=32mm,columnsep=20pt]{geometry} % Document margins
\usepackage{multicol} % Used for the two-column layout of the document
\usepackage[hang, small,labelfont=bf,up,textfont=it,up]{caption} % Custom captions under/above
\usepackage{booktabs} % Horizontal rules in tables
\usepackage{float} % Required for tables and figures in the multi-column environment - they need to
\usepackage[colorlinks=true, urlcolor=green, pdfborder={0 0 0}]{hyperref} % For hyperlinks in the PDF

\usepackage{paralist} % Used for the compactitem environment which makes bullet points with less space
\usepackage[linesnumbered,algoruled]{algorithm2e} % For algorithms 
\usepackage{enumitem}

\usepackage{subfig}

\usepackage{setspace, ifdraft, blindtext} % Spacing in draft versions of papers

\ifdraft{\usepackage{endfloat}}{}

\usepackage{changes}

\colorlet{Changes@Color}{red}

\usepackage[displaymath, mathlines]{lineno}

\usepackage[numbers]{natbib}
\setcitestyle{authoryear, open={(}, close={)} }

\renewcommand{\vec}[1]{\mathbf{#1}}
\def\ie{\rm{{\it i.e.}\ }}
\def\eg{\rm{{\it e.g.}\ }}

\newtheorem{theorem}{Theorem}

\newtheorem{proposition}[theorem]{Proposition}
\newtheorem{corollary}{Corollary}

\theoremstyle{definition}

\theoremstyle{remark}

\graphicspath {{img/}}
\journal{Unknown}
\begin{document}

\ifdraft{\listofchanges[style=list]}{}

\begin{frontmatter}

\title{\small{\color{red} Please cite as: Garbuno-Inigo, Alfredo, Francisco
Alejandro DiazDelaO, and Konstantin M. Zuev. Transitional annealed adaptive
slice sampling for Gaussian process hyper-parameter estimation.{\emph
International Journal for Uncertainty Quantification}, 6.4
(2016).}\\ \large{Transitional annealed adaptive slice sampling for Gaussian
process hyper-parameter estimation}}

\author[liv]{A. Garbuno-Inigo \corref{cor1}}
\ead{agarbuno@liv.ac.uk}

\author[liv]{F. A. DiazDelaO }

\author[caltech]{K. M. Zuev}

\cortext[cor1]{Corresponding author}

\address[liv]{Institute for Risk and Uncertainty, School of Engineering,
 University of Liverpool\\ Brownlow Hill, Liverpool, L69 3GH, United Kingdom}

\address[caltech]{Department of Computing and Mathematical Sciences, California Institute of Technology, Pasadena, CA 91125, USA.}

\begin{abstract}

Surrogate models have become ubiquitous in science and engineering for their
capability of emulating expensive computer codes, necessary to model and
investigate complex phenomena. Bayesian emulators based on Gaussian processes
adequately quantify the uncertainty that results from the cost of the original
simulator, and thus the inability to evaluate it on the whole input space.
However, it is common in the literature that only a partial Bayesian analysis is
carried out, whereby the underlying hyper-parameters are estimated via 
gradient-free optimisation or genetic algorithms, to name a few methods. On the
other hand, maximum a posteriori (MAP) estimation could discard important
regions of the hyper-parameter space. In this paper, we carry out a more
complete Bayesian inference, \replaced{that combines}{ through combining} Slice
Sampling with some recently developed Sequential Monte Carlo samplers. The
resulting algorithm improves the mixing in the sampling through 
delayed-rejection, the inclusion of an annealing scheme akin to Asymptotically
Independent Markov Sampling and parallelisation via Transitional Markov Chain
Monte Carlo. Examples related to the estimation of Gaussian process 
hyper-parameters \deleted{, as well as examples applied in other contexts of
Bayesian inference,} are presented. For the purpose of reproducibility, further
development, and use in other applications, the code to generate the examples in
this paper is freely available for download at
\url{http://github.com/agarbuno/ta2s2_codes}.

\end{abstract}


\end{frontmatter}

% ----------------------------------------------------------------------------
\section{Introduction} 

The use of computationally expensive computer codes is widespread in science and
engineering in order to simulate and investigate complex phenomena. Such codes,
also referred to as {\it simulators}, often require intensive use of
computational resources that allows their use in contexts such as optimisation,
uncertainty analysis and sensitivity analysis \citep{Forrester2008a,
Kennedy2001}. For this reason, surrogate models are needed to efficiently
approximate the output of demanding simulators and enable efficient exploration
and exploitation of the input space. In this context, Gaussian processes are a
common choice to build statistical surrogates -also known as {\it emulators}-
which account for the uncertainty that stems from the inability to evaluate the
original model in the whole input space. Gaussian processes are able to fit
complex input/output mappings through a non-parametric hierarchical structure.
Common applications of Gaussian processes are found\deleted{, amongst many other
areas,} in Machine Learning \citep{Rasmussen2006}, Spatial Statistics
\citep{Cressie1993a} (with the name of Kriging), likelihood-free Bayesian
Inference \citep{Wilkinson2014}, Genetics \citep{Kalaitzis2011} and Engineering
\citep{DiazDelaO2011}\added{, amongst many other areas.}

Building an emulator requires the simulator to be run a repeated number of
times, but due to its computational complexity only a limited amount of
evaluations is available. This cost often translates in an inadequate
explanation of the uncertainty of the model parameters by uni-modal probability
distributions. \added{This phenomenon is often encountered in the form of low
signal-to-noise ratio which translates into multi-modal distributions of the
parameters \citep{Warnes1987}. Although it is not pathological in terms of the
Gaussian process predictive performance, this might lead to solutions with no
significant meaning for the parameters of the surrogate. The Bayesian treatment
of the inference problem often alleviates this concern by means of the prior,
however multi-modality cannot be fully discarded \citep{Andrianakis2012}}. In
\replaced{this setting, where the information available to train a Gaussian
process is limited,}{ these cases,} one is able to acknowledge all uncertainties
related to the modelling assumptions by resorting to Model Uncertainty Analysis
\citep{Draper1995}. More specifically, {\it hierarchical modelling} should be
considered in the model formulation. By doing so, the analyst is capable of
accounting for structural uncertainty, \replaced{which can be considered
either}{ and which be considered} as a continuous or discrete construct
\citep[see][\S 4]{Draper1995}. In Gaussian processes models, continuous
structural uncertainty can be incorporated through the Bayesian paradigm. In
this work, the approach for dealing with different models which are originated
from multi-modal distributions of the model\deleted{'s} parameters is to employ
samplers specifically built for such scenarios.

The implementation of a Gaussian process emulator requires a training phase.
This involves the estimation of the parameters of the Gaussian process, referred
to as {\it hyper-parameters}. The selection of the hyper-parameters is usually
done by Maximum Likelihood estimates (MLE) \citep{Forrester2008a}, Maximum a
Posteriori estimates (MAP) \citep{Oakley1999, Rasmussen2006}, or by sampling
from the posterior distribution \citep{Williams1996} in a fully Bayesian manner.
It is frequently the case that estimating the hyper-parameters depends on
maximising a complex, multi-modal function. In this scenario, traditional
optimisation routines \citep{Nocedal2004} are not able to guarantee global
optima when looking for the MLE or MAP, and a Bayesian treatment becomes a
suitable option to account for all the uncertainties in the modelling. In the
literature however, it is common that either the MLE or MAP alternatives are
preferred \citep{Kennedy2001, Gibbs1998} due to the numerical burden of
maximising the likelihood function or because it is assumed that Bayesian
integration (for example, through Markov chain Monte Carlo methods) is
prohibitive. Although these are strong arguments in favour of point estimates,
in high-dimensional applications it is difficult to assess if the number of runs
of the simulator is sufficient to produce robust hyper-parameters. This is
usually measured with prediction-oriented metrics such as the root-mean-square
error (RMSE) \citep{Kennedy2001a}. The use of this metric ignores the
uncertainty and risk assessment of choosing a single candidate for the 
hyper-parameters through an inference process with limited data. In order to
account for such uncertainty, numerical integration should be performed.
However, methods such as quadrature approximation become quickly infeasible as
the number of dimensions increases \citep{Kennedy2001}. Therefore, a suitable
approach is to perform Monte Carlo integration \citep{MacKay1998}. This allows
to approximate any integral by means of a weighted sum, given a sample from the
{\it correct} distribution.

The Gaussian process model specification yields no direct simulation of its
hyper-parameters by means of standard distributions (Gaussian, uniform or
exponential, to name a few). In order to be able to approximate the related
integrals, Markov chain Monte Carlo (MCMC) provides the proper statistical tool
to generate a desired sample. Nonetheless, the canonical sampling schemes, such
as Metropolis-Hastings or Gibbs sampling, might no be appropriate for 
multi-modal distributions \citep{Neal2001, Hankin2005}. This limitation is
originated by the tuning of the proposal distribution, the function used to
generate samples. If it is not correctly specified, the sampling space might not
be properly explored. The efficiency of the sampler should balance the ability
to move freely through the sampling space as well as to generate candidates
according to the regions where the probability mass is concentrated. One
of the most recognised concerns in simulation is to avoid {\it Random Walk}
behaviour, since it delays the stationarity state achieved by the chain and
limits its exploration capabilities \citep{Neal1993}. An optimal tuning of the
proposal distribution in high-dimensional spaces with intricate correlation
among sets of variables can turn into a demanding task yielding MCMC samplers
expensive \citep{Ching2007}. One alternative in Gaussian processes, amongst
other probabilistic applications, is to resort to the Hybrid Monte Carlo
(HMC) sampler, as it can avoid Random Walk behaviour at \added{the} expense of
additional computations needed for the gradient of the posterior
\citep{Neal1998a,Williams1996}. However, there is no guarantee that multi-modal
distributions can be sampled thoroughly by HMC \citep{Neal2011a} and in Gaussian
process models the gradients are not available in $\mathcal{O}(n)$ operations.

This paper proposes a sampling scheme for \replaced{multi-modal distributions}{
global optimisation} based on two principles. Firstly, the concept of
\textit{crumb} introduced by \cite{Neal2003a} for a multivariate adaptive slice
sampler. Secondly, the ideas by \cite{Zuev2013} on how to simulate local
approximations to the solution based on a sequence of nested subsets as in
Stochastic Subset Optimisation \citep{Taflanidis2008,Taflanidis2008a}. The use
of delayed rejection in the Asymptotically independent Markov sampler
\citep{Garbuno2015} has proven to enhance the mixing capabilities of the
algorithm in highly correlated probability models. To our knowledge, coupling
the adaptive slice sampling algorithm with a sequential sampler has not been
explored previously. This presents an opportunity to develop efficient sampling
algorithms for multi-modal distributions. The main advantage of the proposed
scheme is that it requires little tuning of parameters as it automatically
learns the sequence of temperatures for an annealing schedule, as opposed to
being tuned by trial and error \citep{Birge2012}. The approximation set
simulated in the previous level can be exploited further as it provides the
crumbs needed for the sampling in the next annealing level, leading the
simulation to appropriate regions of the sampling space. Additionally, embedding
the sampler with the Transitional Markov chain Monte Carlo method
\citep{Ching2007} results in an algorithm that can be run in a cluster of cores,
if available. By using the proposed Transitional Annealed Adaptive Slice
Sampling (TA$^2$S$^2$) algorithm to sample the hyper-parameters of a Gaussian
process, the resulting emulator is built taking into account both a
probabilistic and computationally efficient perspective. The probabilistic
strategy to treat the problem in a Bayesian manner accounts for the
uncertainty that stems from the unknown parameters. This adds a layer of
structural uncertainty to the model. Additionally, model uncertainty is
accounted for by adding numerical stabilisation measures in the Gaussian process
model as in \cite{Ranjan2011, Andrianakis2012} in a fully Bayesian framework.

The paper is organised as follows. In Section \ref{sec:gps}, a brief
introduction to the Bayesian treatment of Gaussian processes, as well as related
numerical stabilisation procedures, are presented. Section \ref{sec:ss} briefly
reviews the concepts of Slice Sampling and Adaptive Slice Sampling. Section
\ref{sec:taass} presents the proposed algorithm with the concepts discussed in
the previous sections, as well as the extensions needed for a parallel
implementation. In Section \ref{sec:exp}, some illustrative examples are used to
discuss the efficiency and robustness of the proposed algorithm. Concluding
remarks are presented in Section \ref{sec:conc}.

% ----------------------------------------------------------------------------
\section{The Gaussian process model} \label{sec:gps}

Let the real-valued function $\eta: \mathbb{R}^p \rightarrow \mathbb{R}$
represent the underlying input/output mapping of a \replaced{computer model}{
simulator}. Let $X = \{ \vec{x}_1, \ldots, \vec{x}_n \} $ be the set of {\it
design points}, that is, the set of selected points in the input space, where
$\vec{x}_i \in \mathbb{R}^p$ denotes a given input configuration. Let $ \vec{y}
= \{ y_1, \ldots, y_n \}$ be the corresponding set of outputs $y_i =
\eta(\vec{x}_i)$, such that each pair $(\vec{x}_i, y_i)$ denotes a {\it training
run}. The emulator is assumed to be an interpolator for the training runs, \ie
$y_i = \tilde{\eta}(\vec{x}_i)$ for all $i= 1, \ldots, n$, where the tilde
denotes approximation. This omits any random error in the output of the computer
code in the observed simulations, for which the simulator is said to be
deterministic. If a \added{fully parametrised} Gaussian process prior is assumed
for the outputs of the simulator, then the set of design points has a joint
Gaussian distribution\replaced{. The general assumption is that the simulator
satisfies the statistical model for the output with the following structure}{,
and the output satisfies the structure}

\begin{align}
\eta(\vec{x}) = h(\vec{x})^\top \boldsymbol\beta + Z(\vec{x}|\sigma^2, \boldsymbol\phi),
\end{align}

\noindent where $h(\cdot)$ is a vector of known basis (location) functions of
the input, $\boldsymbol\beta$ is a vector of regression coefficients, and
$Z(\cdot|\sigma^2, \boldsymbol\phi)$ is a Gaussian process with zero mean and
covariance function

\begin{align}
\text{cov}(\vec{x}, \vec{x}'| \sigma^2, \boldsymbol\phi) = \sigma^2 \, k (\vec{x}, \vec{x}'|
\boldsymbol\phi), \label{eq:correlation}
\end{align}

\noindent where $\sigma^2$ is the signal noise and $\boldsymbol\phi \in
\mathbb{R}^p_+$ denotes the {\it length-scale} parameters of the correlation
function $k(\cdot, \cdot)$. \added{The hyper-parameters of the Gaussian process
are therefore $ \boldsymbol \theta = (\boldsymbol \beta, \sigma^2, \boldsymbol
\phi)$. For practical simplicity it is commonly assumed that $h(\vec{x}) = 0$.
This allows to perform predictions and quantify the underlying uncertainty by
relying completely in the covariance function to capture the dependencies among
training runs. Thus $\boldsymbol\beta$ is dropped out of the discussion in
the remainder.} Note that for a pair of design points $(\vec{x},\vec{x}')$, the
function $k(\cdot, \cdot | \boldsymbol\phi)$ measures the correlation between
$\eta(\vec{x})$ and $\eta(\vec{x}')$ based on their respective input
configurations. The correlation function is capable of measuring how close
different input configurations are, such that related inputs produce related
outputs in the simulator. The base of such measure is related to the Euclidean
distance in such a way that it weights differently each input variable. In this
work, the squared-exponential correlation function has been chosen due to its
tractability, namely

\begin{align}
k(\vec{x},\vec{x}'|\boldsymbol\phi) = \exp \left\lbrace - \frac{1}{2} \displaystyle \sum_{i=1}^p
\frac{(x_i - x'_i)^2}{\phi_i}\right\rbrace. \label{eq:covariance-func}
\end{align}

It is important to note that other authors
\citep{Neal1998a,Rasmussen2006,Murray2009} prefer the $\phi^2_i$
parameterisation in the denominator. In our case, it is more natural to use a
linear term, given that the  length-scale parameters are restricted to lie in
the positive orthant. The linear terms can be interpreted as weights in the norm
used to measure closeness and sensitivity to changes in each dimension.

As a consequence of the Gaussian process prior, the joint probabilistic model
for the vector of outputs $\vec{y}$, given the \added{hyper-}parameters
\deleted{$\boldsymbol\beta,$}$\sigma^2$, $\boldsymbol\phi$ and the design points
$X$, can be written as

\begin{align}
{\bf y} | X, \sigma^2, \boldsymbol\phi \sim \mathcal{N}( 0, 
\sigma^2 \, K ), \label{eq:normal}
\end{align}

\noindent where \deleted{$H$ is the {\it design matrix} whose rows are the inputs
$h(\vec{x}_i)^\top $ and }$K$ is the correlation matrix with elements $K_{ij} =
k(\vec{x}_i, \vec{x}_j| \boldsymbol\phi)$ for all $i,j = 1, \ldots, n$.

\subsection{Hyper-parameter marginalisation}
 
The \added{hyper-}parameters of the Gaussian process emulator are commonly
unknown before the training phase, which adds uncertainty to the surrogate. A
common practice in the literature is to fix them to their Maximum Likelihood
value. Though it has been widely accepted, this approach does not entirely treat
the emulator as a probabilistic model and uncertainty quantification through it
might become limited. On the other hand, if one acknowledges the parameters as
random variables, robust estimators can be built through numerical
integration \replaced{by}{.  By} marginalising them via samples generated from
their posterior distribution. \deleted{This allows to incorporate all possible
configurations of the surrogate in light of the evidence shed by the training
runs.} \replaced{Predictions }{After this process, predictions} for $y^*$, given
a  non-observed configuration $\vec{x}^*$, can be performed using all the
evidence provided by the available data $\mathcal{D} = ({\bf y}, X)$, exploiting
the posterior distribution of the hyper-parameters, namely

\begin{align}
p(y^*|\vec{x}^*, \mathcal{D}) = \int_\Theta p(y^*|\vec{x}^*, \mathcal{D}, \boldsymbol\theta) \,
p(\boldsymbol\theta | \mathcal{D}) \, d\boldsymbol\theta \label{eq:pred_post},
\end{align}

\noindent where $\boldsymbol\theta = (\sigma^2, \boldsymbol\phi)$ denotes the
vector of hyper-parameters. Note that the model used for predictions $y^*$,
given the data and $\boldsymbol\theta$, is a Gaussian random variable, which
makes the model closed under the Gaussian family \citep[see][]{Oakley1999}.
\deleted{Additionally, note that each possible specification of
$\boldsymbol\theta$ is a realisation of a Gaussian random variable, by which
referring to $\boldsymbol\theta$ as the Gaussian process' hyper-parameters is
appropriate.}

\added{By means of a sample from the posterior distribution, the mean and
covariance functions of the predictive posterior can be written through a
mixture of Gaussians \citep{Garbuno2015} as}

\deleted{Due to its computational complexity, numerical integration of
\eqref{eq:pred_post} is often avoided. Instead, it is commonly assumed that the
MLE of the likelihood}

\ifdraft{\begin{color}{red}
\begin{align}
\mathcal{L}(\boldsymbol\theta) = p(\vec{y}|X,\boldsymbol\beta, \sigma^2, \boldsymbol\phi) 
\end{align}
\end{color}}{}
\deleted{or the MAP estimate from the posterior distribution }

\ifdraft{\begin{color}{red}
\begin{align}
p(\boldsymbol\theta | \mathcal{D}) \propto p({\bf y}|X,\boldsymbol\beta, \sigma^2, \boldsymbol\phi)
\, \, p(\boldsymbol\beta, \sigma^2, \boldsymbol\phi)
\end{align}
\end{color}}{}
\deleted{are capable of taking into account all the uncertainty in the model. However,
when either the likelihood ?? is a non-convex function or the
posterior ?? is a multi-modal distribution, conventional
optimisation routines may only find local optima. Additionally, there are
degenerate cases when it is crucial to estimate the integral in
\eqref{eq:pred_post} by means of Monte Carlo simulation instead of proposing a
single candidate \citep{Andrianakis2012}. It is possible to approximate the
integrated predictive distribution in \eqref{eq:pred_post} through MCMC as }

\ifdraft{\begin{color}{red}
\begin{align}
p(y^*|\vec{x}^*, \mathcal{D}) \approx \sum_{i=1}^N w_i \,\, p(y^*|\vec{x}^*, \mathcal{D},
\boldsymbol\theta_i) , 
\end{align}
\end{color}}{}
\deleted{ \noindent where $\boldsymbol\theta_i$ is obtained by a sampling scheme
which is capable of managing multi-modal distributions. The coefficients $w_i$
denote the weights of each sample generated. Since each term $p(y^*|\vec{x}^*,
\mathcal{D}, \boldsymbol\theta_i)$ corresponds to a Gaussian density function,
the predictions are made by a mixture of Gaussians. It can be shown
\citep{Garbuno2015} that, if an emulator with input $\vec{x}^*$ and output $y^*$
has a posterior density as in \eqref{eq:mix_posterior}, then its mean and
covariance functions are }

\begin{align}
\mu(\vec{x}^*) &= \sum_{i=1}^N w_i \,\, \mu_i(\vec{x}^*), \label{eq:mix_mean}\\
\text{cov}(\vec{x}^*,\vec{x}') &= \sum_{i=1}^N w_i \,\, \left[ (\mu_i(\vec{x}^*) - \mu(\vec{x}^*))
(\mu_i(\vec{x}') - \mu(\vec{x}')) + \text{cov}(\vec{x}^*,\vec{x}' | \boldsymbol\theta_i) \right],
\label{eq:mix_cov}
\end{align}
\noindent where $\mu_i(\vec{x}^*)$ is the expected value of the \replaced{probability
model}{likelihood distribution} of $y^*$ conditional on the hyper-parameters
$\boldsymbol \theta_i$, the training runs $\mathcal{D}$ and the input
configuration $\vec{x}^*$. From equation \eqref{eq:mix_cov} it follows that the
variance (also known as the prediction error) of an untested configuration
$\vec{x}^*$ is

\begin{align}
s^2(\vec{x^*}) &= \sum_{i=1}^N w_i \,\, ( (\mu_i(\vec{x^*}) - \mu(\vec{x^*}))^2 +
s_i^2(\vec{x^*})). \label{eq:mix_sigma}
\end{align} 
\noindent This results in a more robust estimation of the prediction error,
since it balances the predicted error in one sample with how far the
prediction of such sample is from the overall estimation of the mixture.

\subsection{Prior distributions} 

The predictive posterior distribution in equation \eqref{eq:pred_post} requires
the specification of a prior distribution $p(\sigma^2,
\boldsymbol\phi)$\deleted{in equation \eqref{eq:posterior}}. Weak prior
distributions have been used for \replaced{$\boldsymbol\phi$}{
$\boldsymbol\beta$} and $\sigma^2$ \citep{Oakley1999}, namely

\begin{align}
 p(\sigma^2, \boldsymbol\phi) \propto	\frac{p(\boldsymbol\phi)}{\sigma^2},
\label{eq:prior}
\end{align}

\noindent where \deleted{it is assumed that both the covariance and the mean
hyper-parameters are independent. Even more, $\boldsymbol\beta$ and } $\sigma^2$
is assumed to have a non-informative distribution. For the length-scale 
hyper-parameters $\boldsymbol\phi$, the reference prior \citep{Berger1992a,
Berger2009} allows for an objective framework  \added{in which} the uncertainty
of $\boldsymbol\phi$ can be accounted for. It requires no previous knowledge,
such that the training runs are the only source of information for the inference
process. Additionally, the reference prior is capable of ruling out subspaces of
the sample space of the hyper-parameters \citep{Andrianakis2011}, thus reducing
regions of possible candidates in the mixture model \replaced{expressed}{
proposed as} in equations \eqref{eq:mix_mean}, \eqref{eq:mix_cov} and
\eqref{eq:mix_sigma}. This allows for prior distributions without concerns about
expert knowledge of feasible regions for the hyper-parameters. Unless otherwise
stated, the reference prior for Gaussian processes developed by \cite{Paulo2005}
is used in this work. Note however that there are no known analytical
expressions for the derivatives of this prior, which limits its application to
samplers that require first-order information like HMC. Additionally, it is
important to note that there are other possibilities available for the prior
distribution of $\boldsymbol\phi$. Examples of these are the log-normal or 
log-Laplacian distributions, which can be interpreted as a regularisation in
the norm of the parameters. Other alternatives suggest a decaying prior
\citep{Andrianakis2011} or a weakly informative distribution such as a gamma
with appropriate parameters. If needed, elicited priors from experts can also 
be used \citep{Oakley2002a}.

\subsection{Marginalising the nuisance hyper-parameters}

The \deleted{set of} hyper-parameters \replaced{$\boldsymbol \phi$ and
$\sigma^2$}{of the Gaussian process} may be potentially different in terms of
scales and dynamics, as explained \added{previously} in \cite{Garbuno2015}.
\deleted{By using }Gibbs sampling \replaced{might help with this limitation}{
one can possibly cope}, although it is well-known that such sampling scheme can
be inefficient when used for multi-modal distributions in high-dimensional
spaces. The  Metropolis-Hastings sampler exhibits the same problem\replaced{.
For this reason, this work focuses}{, which leads us to focus} on
$\boldsymbol\phi$ to perform the inference based in the correlation function. To
achieve this, \deleted{$\boldsymbol\beta$ and }$\sigma^2$ is considered as a
nuisance parameter and is integrated out from the posterior. \added{This way,
all inference is driven by the length-scale hyper-parameters of the correlation
function. Note that other authors follow different approaches in the inference
problem of Gaussian processes as emulators. For example, \cite{Vernon2010} focus
on the global trend function $ h(\cdot)^\top \boldsymbol \beta$ since it allows
to incorporate expert knowledge on the computer model being emulated. That is, a
fully parametrised Gaussian process is considered}. The \replaced{probability
model}{ joint distribution} of the training runs and the prior distribution of
the hyper-parameters (equations \eqref{eq:normal} and \eqref{eq:prior}) allow to
identify an \deleted{Gaussian -}inverse-gamma model \replaced{for}{ of
$\boldsymbol\beta$ and} $\sigma^2$, which after integration, can be shown to
yield the integrated posterior distribution

\begin{align}
p(\boldsymbol\phi | \mathcal{D}) \propto p(\boldsymbol\phi) \, (\hat{\sigma}^2)^{-\frac{n-p}{2}} \,
|K|^{-\frac{1}{2}}, \label{eq:integrated_post}
\end{align}
\noindent where 

\begin{align}
\hat{\sigma}^2 &= \frac{\vec{y}^\top \, K^{-1} \, \vec{y}}{n-1}
\end{align}
\deleted{\noindent and}
\ifdraft{\begin{color}{red}
\begin{align}
\hat{\boldsymbol\beta} &= (H^\top K^{-1} H)^{-1}H^\top K^{-1} \vec{y}
\end{align}
\end{color}}{\noindent}
\noindent is an estimator of the signal noise $\sigma^2$ \deleted{and
regression coefficients $\boldsymbol\beta$} \cite[see][for further
details]{Oakley1999}. Finally, the predictive distribution conditioned on the
remaining set of hyper-parameters, $\boldsymbol\phi$, follows a t-student
distribution with mean and correlation functions

\begin{align} \mu(\vec{x}^*| \boldsymbol\phi) =& \,\, t(\vec{x}^*)^\top  K^{-1}
\vec{y} , \\ \text{corr}(\vec{x}^*,\vec{w}^*|\boldsymbol\phi) =& \,\,
k(\vec{x}^*, \vec{w}^* | \boldsymbol\phi) - t(\vec{x}^*)^\top K^{-1} \,
t(\vec{w}^*) , \end{align} 

\noindent where $\vec{x}^*$, $\vec{w}^*$ denote a pair of untested
configurations and $t(\vec{x}^*)$ denotes the vector obtained by computing the
covariance of a new input configuration with every design point available for
training $t(\vec{x}) = (k(\vec{x},\vec{x}_1|\boldsymbol\phi), \ldots, k(\vec{x},
\vec{x}_n| \boldsymbol\phi))^\top $. Note that both the mean and the correlation
of unseen input configurations depend solely on the correlation function 
hyper-parameters $\boldsymbol\phi$. \cite{MacKay1996} has previously discussed
the treatment of nuisance parameters and integrated posteriors. In the context
of Gaussian process emulators, this discussion allows to \deleted{reduce the
dimensionality of the problem and} overcome the limitations of samplers when
faced to different dynamics posed by each subset of hyper-parameters.

In light of the above, this paper focuses on the correlation function
$k(\cdot,\cdot)$ in equation \eqref{eq:correlation}, since the structural
dependencies of the training runs to allow prediction in the outputs of the
simulator is recovered by it. 
% {\color{red} The key aspect is that the trend function, as its
% hyper-parameter, contains minor structural information of the data, when
% compared to the length-scale hyper-parameters. If this were not the case it
% would prevent the use of integrated likelihoods \citep[see][for further
% discussion]{Berger1999}. Nonetheless, if global trend information is available,
% then an additional effort can be made on eliciting an appropriate mean function
% for the emulator. Such expert knowledge of the simulator would allow the analyst
% to model better the mean function by adding significant terms \cite[see][for a
% detailed discussion]{Vernon2010}.}\footnote{This paragraph has to be deleted.}
\subsection{Numerical stability}

Numerical stability is usually difficult to guarantee when implementing
Gaussian processes. As explored previously by \cite{Andrianakis2012} and
\cite{Ranjan2011}, a term  can be added in the covariance matrix $K$ in order to
preserve diagonal dominancy, that is, to add a {\it nugget} hyper-parameter
$\boldsymbol\phi_\delta$ such that

\begin{align}
K_\delta = K + \boldsymbol\phi_\delta \, I, \label{eq:nugget_cov}
\end{align}
\noindent is positive definite. This results in the stochastic simulator 

\begin{align}
y_i = \eta(\vec{x}_i) + \sigma^2 \, \boldsymbol\phi_\delta, 
\end{align}
where the term $\sigma^2 \, \boldsymbol\phi_\delta$ is used to account for the
variability of the simulator that cannot be explained by the correlation
function. The inclusion of the nugget modifies the posterior distribution, and
possibly adds new modes. In such scenario, the use of multi-modal oriented
samplers is crucial for performing the Monte Carlo approximation of equation
\eqref{eq:pred_post}. In case that the resulting model is assessed as not
appropriate, a regularisation term (an elicited prior) can be added to penalise
\deleted{such} regions \added{of local modes} \citep{Andrianakis2012}.

As previously discussed by \cite{Ranjan2011}, a uniform prior distribution
$U(10^{-12}, 1)$ is \replaced{adopted for the nugget hyper-parameter
$\phi_\delta$}{ used for our framework}. The lower bound guarantees stability in
computations associated with the covariance matrix. The upper bound forces the
numerical noise of the simulator to be smaller than the signal noise of the
emulator itself. By considering the modified correlation matrix in equation
\eqref{eq:nugget_cov}, the previous assumptions regarding the original noise
parameter $\sigma^2$ remain unchanged,  as one can still marginalise it as a
nuisance parameter with the non-informative prior used \citep{DeOliveira2007}.

% ----------------------------------------------------------------------------

\section{Slice sampling \label{sec:ss}  } 

The slice sampling algorithm \citep{Neal2003a} is a method to simulate
a Markov chain of a random variable $\boldsymbol\theta \in \Theta$.
This is done by introducing an auxiliary random variable $u \in \mathcal{U}
\subseteq \mathbb{R}$ and sampling from the joint distribution on the extended
space $\Theta \times \mathcal{U}$. The marginal of $\boldsymbol\theta$ is
recovered by disregarding the values of $u$ in the Markov chain, a consequence
of defining an appropriate conditional distribution for $u$, given
$\boldsymbol\theta$. The samples are generated by an iterative Gibbs sampling
schedule to recover pairs $\{ (\boldsymbol\theta_i, u_i)\}_{i=1}^N$, which
follow the joint density probability distribution

\begin{align}
\pi(\boldsymbol\theta,u) \propto I_{\{u < \pi(\boldsymbol\theta)\}}(\boldsymbol\theta, u),
\end{align}

\noindent where $I_E(\cdot)$ is the indicator function for the set $E \subset
\Theta \times \mathcal{U}$, and $\pi(\boldsymbol\theta)$ is the target
distribution of $\boldsymbol\theta$. Slice Sampling first generates $u$ from the
conditional \added{distribution of} $u\, | \, \boldsymbol\theta$ specified as a
uniform \deleted{distribution}on the interval $(0, \pi(\boldsymbol\theta))$. It then
samples $\boldsymbol\theta$, conditioned in $u$ from a uniform distribution in
the \textit{slice} defined by the set

\begin{align}
S_u = \{ \boldsymbol\theta : u < \pi(\boldsymbol\theta) \}. \label{eq:slice_set}
\end{align}
Since the marginal satisfies $ \int_0^{\pi(\boldsymbol\theta)} \pi(
\boldsymbol\theta , u ) \, du = \pi(\boldsymbol\theta) $, samples from the
target distribution can be recovered by disregarding the auxiliary component of
the joint samples. If the target distribution is a non-normalised probability
density $f(\boldsymbol\theta)$ then the joint distribution can be written as

\begin{align}
\pi(\boldsymbol\theta,u) = \frac{1}{Z} \, I_{\{u < f(\boldsymbol\theta) \}}(\boldsymbol\theta, u),
\end{align}
where \deleted{$Z$ is the normalising constant for the target, that is} $Z =
\int_\Theta f(\boldsymbol\theta) \, d\boldsymbol\theta$ and the previous
considerations for the marginal of $\boldsymbol\theta$ follow. In the context of
Gaussian processes it should be noted that floating-point underflows are common
due to ill-conditioning of the matrix $K$ in equation
\eqref{eq:integrated_post}. Thus, in order to compute stable evaluations of the
target distribution in Slice Sampling, it is preferable to evaluate the negative
logarithm of the target density. In such case, equation \eqref{eq:slice_set} can
be computed as stated in the following proposition.

\begin{proposition} \label{prop:slice}

Given the state of the Markov chain $\boldsymbol\theta_0$, the uniform
distribution for the next candidate has support in the slice given by 

\begin{align}
S_{\boldsymbol\theta_0} = \{ \boldsymbol\theta : z > \mathcal{H}(\boldsymbol\theta) \}, \label{eq:slice}
\end{align}
where $\mathcal{H}(\cdot)$ denotes the negative logarithm of the target density
and $z = \mathcal{H}(\boldsymbol\theta_0) + e$, with $e$ distributed as an
exponential random variable with mean equal to 1. 
\end{proposition}

\begin{proof} 

The result follows from the fact that for a given state
$\boldsymbol\theta_0$ of the Markov chain, the auxiliary uniform random variable
defining the slice can be written as the product $u \, f(\boldsymbol\theta_0)$
with $u$ uniformly distributed in the interval $(0,1)$. Thus, the slice is
defined as

\begin{align}
S_{\boldsymbol\theta_0} & = \{ \boldsymbol\theta : u \, f(\boldsymbol\theta_0) < f(\boldsymbol\theta) \} \nonumber \\
& = \{ \boldsymbol\theta : - \log ( f(\boldsymbol\theta_0) ) -\log(u)> -\log( f(\boldsymbol\theta) ) \} \nonumber \\
& = \{ \boldsymbol\theta : \mathcal{H}(\boldsymbol\theta_0) + e> \mathcal{H}(\boldsymbol\theta) \} ,
\end{align}
where it is easy to prove that $e = - \log(u)$ is distributed as an exponential
random variable with mean 1 and $ \mathcal{H} (\cdot) $ denotes the negative
logarithm of the target density.
\end{proof}

The main concern when implementing Slice Sampling is the ability to sample
uniformly from the slice. In one-dimensional applications, the slice can be
defined by several methods. The canonical example is a stepping-out and
shrinkage procedure which aims to extend an initial interval until the edges are
outside the slice and then shrinking it as samples generated from the interval
are rejected \citep[see][for further details]{Neal2003a}.

\subsection{Adaptive slice sampling} \label{subsec:ass}

For multivariate distributions, the concept of the slice extends naturally.
However, methods based on intervals (\eg the stepping-out and shrinking
procedure) slow dramatically as the dimension of the problem increases. This is
due to the generalisation of intervals as hyper-rectangles in $\mathbb{R}^p$ and
the need to compute the target function for each vertex a repeated number of
times along the expansion and shrinkage of the boundaries. For Gaussian process
emulators, the task of evaluating the target density becomes expensive\added{, a
consequence of the non-parametric nature of the model and the computational cost
of evaluating equation \eqref{eq:pred_post}.} \replaced{If}{ if} multiple
evaluations are needed for \replaced{the construction of }{ each state in} the
Markov chain \added{either because of a high rejection rate, difficult
characterisation of the slice or if longer chains are required, simulation by
MCMC with slice sampling becomes computationally expensive and inefficient.}
Therefore, other alternatives are preferable.

This work employs a framework proposed by \cite{Neal2003a} for adaptive
slice sampling in multivariate applications. The key idea is the use of the
information provided by the rejected samples in order to lead the future
generation of a candidate inside the slice. In this framework, the evidence
gathered by the rejected candidates is referred to as \textit{crumbs}, as they
will be ``followed'' towards the slice. \added{For a more detailed discussion
see \citep{Neal2003a}.} \deleted{The algorithm starts with a given candidate
$\boldsymbol\theta_0$ and a slice defined by a uniform random variable $u$ in
the interval $(0, f(\boldsymbol\theta_0))$. The first crumb, denoted by
$\varsigma_1$, is drawn from a density distribution conditioned on the slice and
the current point, that is $g_1(c | u,\boldsymbol\theta_0)$. Let
$\boldsymbol\theta_1^*$ be the first candidate drawn from the proposal density
$h_1(\boldsymbol\theta^* | u, \varsigma_1) \propto g_1(\varsigma_1 |
u,\boldsymbol\theta^*)$. If the candidate is on the slice then it becomes the
next state of the Markov chain. If it is rejected, a second crumb is generated
from a density $g_2(c | \boldsymbol\theta_0, u, \varsigma_1,
\boldsymbol\theta_1^*)$, which may depend on the previous rejected sample and
crumb, as well as on the current state of the chain and slice level. The next
candidate is then drawn from a density $ h_2( \boldsymbol\theta^* | u,
\varsigma_1, \boldsymbol\theta_1^*, \varsigma_2) \propto g_1(\varsigma_1 |
u,\boldsymbol\theta^*) \, g_2(\varsigma_2 | \boldsymbol\theta^*, u, \varsigma_1,
\boldsymbol\theta_1^*) $. If the second candidate is inside the slice then a
Gibbs update is performed to the slice level and the procedure repeats with the
new state of the Markov chain. However, if the candidate is rejected then a
third crumb is generated following the procedure described above, conditioned on
the previous crumbs generated, the current state of the chain and the rejected
samples. It can be proved \citep{Neal2003a} that this procedure generates a
Markov chain leaving the target distribution $f(\boldsymbol\theta)$ invariant,
leading to proposal distributions more concentrated as the number of crumbs
increases. This is due to the fact that the procedure can be interpreted as
generating samples from pseudo-posterior distributions based on pseudo-data (the
crumbs).}

% ----------------------------------------------------------------------------
\section{Transitional annealed adaptive slice sampling} \label{sec:taass}

\added{As previously stated,} in order to marginalise the posterior predictive
distribution in equation \eqref{eq:pred_post}, Monte Carlo integration is
usually performed when aiming at a fully Bayesian treatment \added{of Gaussian
process surrogates}. This is usually done by Hybrid Monte Carlo
\citep{Neal1998a,Williams1996} which is capable of suppressing the Random Walk
behaviour of traditional MCMC methods. Nonetheless, the tuning of this kind of
algorithm is problem-dependent and expert knowledge is crucial for an optimal
sampling schedule. The development of Elliptical Slice Sampling
\citep{Murray2009} provides a framework for the simulation of the 
hyper-parameters of a Gaussian process with little tuning required from the
analyst \citep{Murray2010}. However, this is only applicable when the posterior
predictive distribution for the hyper-parameters is of the form

\begin{align}
p(\boldsymbol\theta| \mathcal{D}) \propto \mathcal{N}( l(\boldsymbol\theta) | \boldsymbol\mu, \Sigma) \, p(\boldsymbol\theta),
\end{align}
where $p(\boldsymbol\theta)$ denotes the prior distribution, $\mathcal{N}(\cdot
| \cdot, \cdot)$ is a Gaussian distribution and $l(\cdot)$ is a latent variable
that depends on the hyper-parameters. As it can be seen from the integrated
posterior in equation \eqref{eq:integrated_post}, this form is not valid in the
formulation followed. The difference \replaced{stems from $\sigma$ being
considered a nuisance parameter and the prior considered for
the length-scales}{ is a full regression model assumed in the mean function} of
the Gaussian process. \deleted{Note that by marginalising the regression coefficients and
signal noise of the Gaussian process, one aims at reducing the dimensionality of
the problem as well as coping with different scales for different sets of hyper-
parameters.}

In this setting, we propose Transitional Annealed Adaptive Slice Sampling
(TA$^2$S$^2$), which can also be used in other applications of Bayesian
inference and Stochastic optimisation. Based on Asymptotically
Independent Markov Sampling \citep{Beck} we formulate the \added{sampling}
problem to be solved as reminiscent of simulated annealing. \replaced{The
objective is to sample from intermediate posterior distributions
$p_k(\boldsymbol\phi | \mathcal{D})$ that eventually converge to the true
posterior. This is done by tempering the posterior distribution by means of a
monotonically decreasing sequence of temperatures $\tau_k$ converging to 1. Let
$\{p_k(\boldsymbol\phi | \mathcal{D})\}_{k = 1}^{\infty}$ be the sequence of
density distributions in the annealing schedule such that }{ We want to find the
solution to the optimisation problem}

\ifdraft{\begin{color}{blue}}{\begin{color}{black}}
\begin{align}
p_k(\boldsymbol\phi | \mathcal{D}) &\propto p(\boldsymbol\phi | \mathcal{D})^{1/\tau_k} = \exp \left\lbrace -
\mathcal{H}(\boldsymbol\phi | \mathcal{D} )/\tau_k \right\rbrace, \label{eq:target_level}
\end{align} 
\end{color}
\ifdraft{\begin{color}{red}
\begin{align}
\min_{\boldsymbol\phi \in \Phi} \mathcal{H} ( \boldsymbol \phi | \mathcal{D}),
\end{align}
\end{color}}{}
\noindent where $\mathcal{H} ( \boldsymbol \phi | \mathcal{D})$ denotes the
negative integrated log-posterior distribution of the length-scale 
hyper-parameters, given the set of training runs $\mathcal{D}$. 
\deleted{Let the set of optimal solutions to the problem above be denoted by}

\ifdraft{\begin{color}{red}
\begin{align}
\Phi^* = \left\lbrace \boldsymbol\phi \in \Phi \,:\, \boldsymbol\phi = \arg \min_{\boldsymbol\phi
\in \Phi} \, \mathcal{H}(\boldsymbol\phi | \mathcal{D} ) \right\rbrace , \label{eq:opt_set}
\end{align}
\end{color}}{}
\deleted{with $|\Phi^*| \geq 1$. Note that this formulation of the problem
allows for the existence of multiple global optimisers.} The algorithm provides
a sequence of nested subsets $\Phi_{k+1} \subseteq \Phi_{k}$ converging to the
set of \replaced{posterior samples}{ optimal solutions} $\Phi^*$. \deleted{That
way, if the algorithm is terminated in an intermediate annealing level, a sample
of local approximations to the global minimiser can be recovered. The annealing
is done by tempering the target distribution using a sequence of density
functions as in simulated annealing.} The temperature is learned through an
automatic mechanism to determine the sequence of distributions. \deleted{Let
$\{p_k(\boldsymbol\phi | \mathcal{D})\}_{k = 1}^{\infty}$ be the sequence of
density distributions in the annealing schedule such that}

\ifdraft{\begin{color}{red}
\begin{align}
p_k(\boldsymbol\phi | \mathcal{D}) &\propto p(\boldsymbol\phi | \mathcal{D})^{1/\tau_k} = \exp \left\lbrace -
\mathcal{H}(\boldsymbol\phi | \mathcal{D} )/\tau_k \right\rbrace,
\end{align} 
\end{color}}{}
\deleted{for a sequence of monotonically decreasing temperatures $\tau_k$
converging to zero.} By construction, the sample in the first level of annealing
is distributed uniformly on a \textit{practical support} of the sampling space
\citep[see][for more a detailed discussion]{Katafygiotis2007}. For the limiting
case, the samples are uniformly distributed in the \replaced{support of the
posterior density}{ set of optimal solutions}. Both these observations can be
summarised by

\begin{align}
\lim_{\tau \rightarrow \infty} p_\tau(\boldsymbol\phi| \mathcal{D}) &= U_\Phi(\boldsymbol\phi),
\label{eq:meta_prior}\\
\lim_{\tau \rightarrow 1} p_\tau(\boldsymbol\phi| \mathcal{D}) &= U_{\Phi^*}(\boldsymbol\phi),
\end{align}
where $U_A(\boldsymbol\phi)$ denotes a uniform distribution over the set $A$ for
every $\boldsymbol\phi \in A$.

\subsection{Annealing at level k} \label{subsec:annealing_trans}

This subsection focuses on the sampling carried out by TA$^2$S$^2$ at the $k$-th
level of the annealing sequence. It is therefore assumed that a sample from
level $k-1$, which is distributed according to $p_{k-1}(\boldsymbol\phi |
\mathcal{D}) $, has already been generated. Let $\boldsymbol\phi_1^{(k-1)},
\ldots, \boldsymbol\phi_N^{(k-1)}$ denote such sample and let $N$ be the sample
size in each annealing level. Following the ideas discussed in Section
\ref{subsec:ass} for Adaptive Slice Sampling, the \textit{crumb} formulation
will be exploited. The samples from the previous level play the role as the
crumbs to be followed to generate candidates from each slice. \added{Thus,
retaining information from the posterior landscape and limiting the
amount of evaluations of the integrated posterior, which can be expensive for a
reasonable number of training runs. }Firstly, note that Proposition 1 implies
the following

\begin{corollary}
The slice defined in the $k$-th annealing level, given the current state of the
Markov chain $\boldsymbol\phi_0$, is given by
\begin{align}
S_{\boldsymbol\phi_0}^k = \{ \boldsymbol\phi : z_k > \mathcal{H}(\boldsymbol\phi | \mathcal{D}) \}
\label{eq:slice_annealing},
\end{align}
where $z_k = \mathcal{H}(\boldsymbol\phi_0 | \mathcal{D}) + e_k$, with $e_k$ an
exponential random variable with mean $\tau_k$. \end{corollary}

% \begin{proof}
% The result follows directly from the definition of the target distribution at the
% annealing level $k$ in equation \eqref{eq:target_level} and as direct
% consequence from Proposition \ref{prop:slice}.
% \end{proof}

As in other Sequential Monte Carlo algorithms \citep{DelMoral2006,Fearnhead2013},
let us define the importance weights of the samples $\boldsymbol\phi_1^{(k-1)},
\ldots, \boldsymbol\phi_N^{(k-1)}$ as

\begin{align}
\omega^{(k-1)}_j &= \frac{p_{k}\left( \boldsymbol\phi^{(k-1)}_j \right)}{p_{k-1}\left(
\boldsymbol\phi^{(k-1)}_j \right)} \propto \exp \left\lbrace - \mathcal{H}\left( 
\boldsymbol\phi^{(k-1)}_j | \mathcal{D} \right) \left( \frac{1}{\tau_k} -
\frac{1}{\tau_{k-1}}\right)\right\rbrace, \label{eq:unweight}\\
\overline{\omega}^{(k-1)}_j &= \frac{\omega^{(k-1)}_j}{\sum_{j=1}^N \omega^{(k-1)}_j},
\end{align}
where $\omega^{(k-1)}_j$ denotes the importance weights and
$\overline{\omega}^{(k-1)}_j$ the normalised importance weights. The weights
allow to measure the importance of each sample as being drawn for
the next annealing level.

The proposal for a new state of the Markov chain, given the current one
$\boldsymbol\phi_0$, is generated as follows. A slice is obtained as in equation
\eqref{eq:slice_annealing} by generating an exponential random variable with
mean $\tau_k$, thus defining the slice $S_{\boldsymbol\phi_0}^k $ for the
current state. A first crumb is randomly selected from the set of past
approximations that lie inside the slice. This means selecting a uniformly
distributed index $j$ from the set

\begin{align}
\mathcal{J} = \left\lbrace j \in \{1, \ldots, N\} \, : \, \boldsymbol\phi_j^{(k-1)} \in
S_{\boldsymbol\phi_0}^k \right\rbrace. 
\end{align}
The points $\boldsymbol\phi_1^{(k-1)}, \ldots, \boldsymbol\phi_N^{(k-1)}$ are
uniformly distributed in the approximation set $\Phi_{k-1}$ and will be used as
markers for the annealing level $k$. If the above index set is empty, there is
evidence of the annealing temperature being decreased too rapidly. A fail-safe
can be used by generating a crumb from a wide Gaussian distribution centred at
the current state of the Markov chain. Additionally, as it is done in other
Sequential Monte Carlo methods \citep{DelMoral2012}, a renewal component can be
added. The renewal is performed as the crumbs are selected from the markers, due
to the fact that relying on the sample from the previous level can lead to bias
in the simulations. To this end, if the index set $\mathcal{J} $ is empty or,
with probability $p_{\text{renew}}$, the crumb $\boldsymbol\varsigma_1$ will be
distributed as 

\begin{align}
\boldsymbol\varsigma_1 \sim \mathcal{N}( \boldsymbol\phi_0, c_0^2 \, \Sigma_k), 
\end{align}
where $c_0$ is a spread parameter associated with the annealing sequence, and
$\Sigma_k$ denotes a covariance matrix at level $k$. Typical choices for the
covariance matrix are the identity matrix $I_{p \times p}$ or a diagonal matrix
$\text{diag}\{d_1, \ldots, d_p\}$ which defines a different scale for each
variable. In order to use a better proposal in terms of scales and correlations
observed along the annealing sequence, we define $\Sigma_k$ as the weighted
covariance matrix from the weighted samples $\{( \overline{\omega}^{(k-1)}_j,
\boldsymbol\phi_j^{(k-1)})\}_{j=1}^N$. As discussed by \citep{Gelman1996}, the
spread parameter is set as $c_0 = 2.38/\sqrt{p}$, since it allows for efficient
transitions in Gaussian steps.

Once the first crumb is drawn, a first candidate $\boldsymbol\xi_1$ is
generated from the appropriate Gaussian distribution 

\begin{align}
\boldsymbol\xi_1 \sim \mathcal{N}(\boldsymbol\varsigma_1, c_0^2 \, \Sigma_k), 
\end{align}
where $c_0$ is a spread parameter for the proposals and $\Sigma_k$ defined as
above. In general, the $i$-th candidate for the next state of the Markov chain
can be generated as

\begin{align}
\boldsymbol\xi_i \sim \mathcal{N}\left(\overline{\varsigma}_i , \left(\frac{c_0}{i}\right)^2 \, \Sigma_k \right),
\end{align}
where $\overline{\varsigma}_i$ is the average of the crumbs generated so far,
as proposed by \citep{Neal2003a}. Note how the generation of new candidates in
the slice is narrower as the candidates are rejected by means of the parameter
$c_0/i$. However, the mean for the Gaussian proposal might not converge to a
point in the slice if the posterior is a multi-modal distribution. To cope with
this limitation, we propose to use a weighted average of the current state and
the crumb centre to enhance the mixing of the sampler. Namely, by sampling the
$i$-th candidate from a Gaussian distribution with mean

\begin{align}
\overline{\varsigma}_i^* &= \alpha_i \, \boldsymbol\phi_0 + (1- \alpha_i) \, \overline{\varsigma}_i
\end{align}
and covariance $ (c_0/i) \, \Sigma_k$. The weight parameter $\alpha_i$ can be
defined in terms of the number of crumbs previously rejected. Since it is
desirable that $\alpha_i \rightarrow 1 $ as $i$ increases, we can define it
either as $\alpha_i = (1 - 1/i)$ or $\alpha_i = (1 - \exp(-i))$. As confirmed by
our experiments, $\alpha_i$ is linearly-dependent on the crumb iteration, since
the exponential behaviour exhibits pronounced decay towards the current state,
causing Random Walk behaviour. The sampling in each annealing level is
summarised in Algorithm \ref{alg:annealing_level}. To avoid cluttered notation,
the conditioning on the design points $\mathcal{D}$ is dropped in the
remainder.

\begin{algorithm}[H]
\SetKwInOut{Input}{Input}
\SetKwInOut{Output}{Output}
\SetKwRepeat{Do}{do}{while}

\Input{\\
\vspace{-1.5mm}
\begin{itemize}[label=$\diamond$]\setlength{\itemsep}{-1mm}
\item $\boldsymbol\phi^{(k-1)}_1, \ldots, \boldsymbol\phi^{(k-1)}_N \sim p_{k-1}(\boldsymbol\phi)$,
generated at previous level;
\item $\boldsymbol\phi_1^{(k)} \in \Phi$, initial state of the chain;
\end{itemize}
\vspace{-1.5mm}
}

\Output{\\
\vspace{-2.5mm}
\begin{itemize}[label=$\diamond$]\setlength{\itemsep}{-2mm}
\item $\boldsymbol\phi^{(k)}_1, \ldots, \boldsymbol\phi^{(k)}_N \sim p_{k}(\boldsymbol\phi)$;
\end{itemize}
\vspace{-1.5mm}
}

\BlankLine

\Begin{
	Compute covariance matrix $\Sigma_k$ from the weighted samples\;
	\For{$i\leftarrow 2$ \KwTo $n-1$}{
		Define slice $S_{\boldsymbol\phi_i}^k$ as in \eqref{eq:slice_annealing} \;
		$l \leftarrow 0 $\; 
		\Do{ $\boldsymbol \xi_i \notin S_{\boldsymbol\phi_i}^k$ }{
			Increase $l$ and generate $u \sim U(0,1)$\; 
			\eIf {$|\mathcal{J}| \neq \emptyset$ or $u < p_{\text{renew}}$ }{
				Choose random $j$ from index set $\mathcal{J}$\; 
				$\varsigma_l = \boldsymbol\phi^{(k-1)}_j$ \;
			}{
				Generate $\varsigma_l \sim \mathcal{N} ( \boldsymbol\phi^{(k)}_i, c_0^2 \, \Sigma_k)$ \;
			}
			Define crumb as $\overline{\varsigma}_l^* = \alpha_l \, \boldsymbol\phi^{(k)}_i + (1-
\alpha_l) \, \overline{\varsigma}_l$ \; 

			Generate candidate $ \boldsymbol \xi_i \sim \mathcal{N} (
			\boldsymbol\phi^{(k)}_i, (c_0/l)^2 \, \Sigma_k)$ \;
			
		}
		Define new state of the chain $\boldsymbol\phi^{(k)}_{i+1} = \boldsymbol\xi_i $\;
	}
}
\caption{TA$^2$S$^2$ at annealing level $k$}\label{alg:annealing_level}
\end{algorithm}

\subsection{Overview of the full sampler}

The algorithm starts with a uniform sample in an admissible space $\Phi$, as
implied by the \textit{meta}-prior distribution in equation
\eqref{eq:meta_prior}. As a second step, the algorithm described in the previous
section is used to generate the samples of the first annealing level, that is
$\boldsymbol\phi_1^{(1)}, \ldots, \boldsymbol\phi_N^{(1)} \sim
p_1(\boldsymbol\phi)$. As mentioned before, this set of points \replaced{allows
to approximate the slices in the next annealing level and the areas where the
posterior mass is concentrated.}{ is the first approximation to the solution of
the optimisation problem}. As the sequence of temperatures converges to 1, we
expect to recover better approximations until \replaced{posterior samples are
generated. That is, }{a desired sample is simulated. Thus}, until a sample
$\boldsymbol\phi_1^{(k^*)}, \ldots, \boldsymbol\phi_N^{(k^*)}$ has been drawn
and is uniformly distributed in the set $\Phi^*$. \replaced{In the next section
we discuss}{We also need to determine a stopping criterion for the sampler, as
well as defining} how to learn the temperature sequence and the overall parallel
implementation achieved by embedding it on a transitional Markov chain schedule.

\subsubsection{Annealing schedule} \label{subsec:annealing-schedule}

The way the temperature sequence is determined is one of the most crucial
aspects of any simulated-annealing-based method. It is clear that if the change
of temperatures is abrupt the markers will degenerate quickly, as observed in
sequential Monte Carlo samplers. On the contrary, if the sequence of
temperatures decreases slowly the actual efficiency of the algorithm is
hindered, since sampling in a sequence of annealing levels is redundant for the
\replaced{generation of posterior sample}{approximation of the solution}.
Setting the temperature sequence beforehand requires prior knowledge of the
overall behaviour of the function $\mathcal{H}(\cdot)$ and the topology around
the set $\Phi^*$, both of which are generally not available.

Following the suggestion by \cite{Zuev2013}, the {\it Effective Sampling Size}
can be used as a measure of degeneracy of the chain in each annealing level
\citep[see][for further discussion]{Zuev2013}. This allows to measure how
similar the $(k-1)$-th and the $k$-th densities are. The effective sample size
can be approximated by

\begin{align}
\hat{n}_{\text{eff}}= \frac{1}{\sum_{i = 1}^N \left( \overline{\omega}^{(k-1)}_j \right) ^2}, 
\end{align}
where $\overline{\omega}^{(k-1)}_j$ is the normalised weight of sample
$\boldsymbol\phi^{(k-1)}_j$. Given the temperature of the previous level is
known, the problem is to determine the temperature of the next one. This is
done by determining a target threshold for $\hat{n}_{\text{eff}} $ in terms of
the size of the simulated set. Thus, given $\gamma \in (0, 1)$, the target
threshold is defined by $\gamma N = \hat{n}_{\text{eff}}$. Rewriting this
expression in terms of the unnormalised sample weights we obtain

\begin{align}
\frac{\sum_{j=1}^N \exp \left\lbrace -2 \mathcal{H} \left( \boldsymbol\phi_j^{(k-1)} \right) \, \left(
\frac{1}{\tau_k} - \frac{1}{\tau_{k-1}}\right) \right\rbrace }{\left( \sum_{j=1}^N \exp
\left\lbrace - \mathcal{H} \left( \boldsymbol\phi_j^{(k-1)} \right) \, \left( \frac{1}{\tau_k} -
\frac{1}{\tau_{k-1}}\right) \right\rbrace\right)^2} = \frac{1}{\gamma N},
\end{align}
which yields an equation for the unknown temperature $\tau_k$. Solving
the equation for $\tau_k$ can be done efficiently by standard numerical techniques
such as the bisection method.

The value of the threshold $\gamma$ affects the overall efficiency of the
annealing schedule. If a value close to zero is chosen, the resulting algorithm
will create few tempered distributions and this will result in poor
approximations. If $\gamma$ is close to 1, then there will be excessive
tempered distributions and redundant annealing levels. As suggested by
\cite{Beck}, and as confirmed by our experiments, a value of $\gamma = 0.5$
delivers acceptable efficiency.

\ifdraft{\subsubsection{{\color{red} Stopping criterion}}}{}

\deleted{In theory, the algorithm should run with $k$ increasing until the
temperature $\tau_{k} \rightarrow 0$, resulting in uniformly
distributed samples in the set $\Phi^*$. However, in practical implementations,
the absolute zero cannot be achieved and a stopping criterion
for the annealing sequence is necessary. In this situation, the sample
coefficient of variation (COV) of the objective function $\mathcal{H}(\cdot)$
can serve as a stopping criterion. Without loss of generality, it is assumed
that the objective $\mathcal{H}(\cdot)$ is a non-negative function
\citep{Zuev2013}. Let $\delta_k$ denote the
sample COV of the simulations $\mathcal{H}(\boldsymbol\phi^{(k)}_1 ),
\ldots, \mathcal{H}(\boldsymbol\phi^{(k)}_N )$, that is }

\ifdraft{\begin{color}{red}
\begin{align}
\delta_k = \frac{\sqrt{\frac{1}{N} \sum_{i = 1}^N \left( \mathcal{H}\left(\boldsymbol\phi_i^{(k)}
\right) - \frac{1}{N} \sum_{j = 1}^N \mathcal{H} \left(\boldsymbol\phi_j^{(k)} \right) \right)^2
}}{ \frac{1}{N} \sum_{j = 1}^N \mathcal{H} \left(\boldsymbol\phi_j^{(k)} \right)}. 
\end{align}
\end{color}}{}
\deleted{As noted in \citep{Garbuno2015} and \citep{Zuev2013} using $\delta_k$
yields a measure of the sensitivity of the objective (the integrated log-
posterior distribution) to the hyper-parameters, in the domain induced by the
annealing temperature $\tau_k$. If the set of approximated solutions
$\boldsymbol\phi_1^{(k)}, \ldots, \boldsymbol\phi_N^{(k)}$ lie in  $\Phi^*$,
then the sample COV will equal 0 exactly, since for all $ j$,
$\mathcal{H}(\boldsymbol\phi_j^{(k)}) = \min_{\boldsymbol\phi \in \Phi^*} \,
\mathcal{H}(\boldsymbol\phi)$. Thus, the stopping criterion is defined when a
threshold of the initial sample COV has been reached, that is}

\ifdraft{\begin{color}{red}
\begin{align}
\delta_k < \alpha \, \delta_0,
\end{align}
\end{color}}{}
\deleted{where $\alpha \in (0,1)$ defines the target threshold attained by the
sampling algorithm in the last annealing level. The value of $\alpha$ directly
affects the accuracy of the approximation achieved by the sampling algorithm. In
the context of Gaussian process emulators, our experiments suggest that
$\alpha = 0.1$ yields good approximations recovering robust configurations
of the surrogate model for error prediction and predictions made through
equations \eqref{eq:mix_mean} and \eqref{eq:mix_sigma}.}

\subsubsection{Parallel Markov chains} 

As described so far, the proposed algorithm can be computationally expensive if
the Markov chain of the samples is drawn sequentially. This is due to the
inversion of a $n\times n$ matrix and related products in equation
\eqref{eq:integrated_post}. Hence, it is desirable to speed up the process of
generating samples in each annealing level. In our context, the inversion of
such matrix is not prohibitive, since we assume that the set of training points
is expensive to acquire, however, a fast sampling algorithm is desired for a
complete Bayesian treatment of the problem. This way we can compensate the
drawbacks associated with an appropriate error estimation by using the emulator
in a Bayesian setting \citep[see][for a discussion]{Kennedy2001}. The idea of
parallelisation comes from an adaptation of the Transitional Markov Chain Monte
Carlo (TMCMC) method \citep{Ching2007} in the context of the annealed adaptive
slice sampling algorithm described previously.

The TMCMC algorithm builds a Markov chain from a target distribution in a
sequential schedule as in Sequential Monte Carlo \citep{DelMoral2007} and
Particle Filtering \citep{Andrieu2010}. That means that $N$ Markov chains are
started, each from the state of an initial Markov chain being drawn from the
prior distribution of the Bayesian inference problem. The key difference is that
the chains are allowed to communicate among each other by a {\it transition}
mechanism that allows to grow each chain differently within the same annealing
level, disregarding poor initial states for certain chains. The length of the
chain is determined by a probability proportional to the importance sampling
weight defined in equation \eqref{eq:unweight}. By doing so, the markers are
automatically selected in the updating sequence and concentrated around the
modes found during the annealing. This improves the mixing of the samples
generated in each annealing level.

Summarising, the proposed TA$^2$S$^2$ algorithm consists of Markov chains
generated as established in Algorithm \ref{alg:annealing_level}, the annealing
temperature being determined empirically by the effective sampling size
described in Section \ref{subsec:annealing-schedule} and stopped whenever
\replaced{the temperature reaches 1}{ the criterion of Section
\ref{subsec:stopping} is met}. The selection of the initial states of the Markov
chains and their growth length is a direct implementation of the TMCMC method
\citep{Ching2007} for Bayesian model updating.

%----------------------------------------------------------------------------
\section{Numerical experiments} \label{sec:exp}

The following examples illustrate the effectiveness and robustness of
TA$^2$S$^2$ when sampling the hyper-parameters of Gaussian process surrogates.
The first example is Franke's function \citep{Haaland2012}, wich can have
challenging features when emulated. The second example is a five-dimensional
model \citep{Nilson1989} which has been previously used to test Gaussian
processes \replaced{meta-models}{ surrogates} \citep{Bastos2009}. The third
example is a ten-dimensional model for the weight of a wing of a light aircraft
\citep{Forrester2008a}. \deleted{Additionally, we present two examples that
illustrate the application of TA$^2$S$^2$ in other contexts. These are related
to Bayesian inference, assuming a scenario where the MAP estimate is not
appropriate and a more robust marginalisation is needed: first, an application
to Bayesian Neural Networks \citep{Beck} and second, an application to
Variational inference where the solution is known to have multiple local optima
\citep{Bishop2006}.} \deleted{Unless otherwise stated, it is assumed that the
global behaviour of the emulator can be fitted by a simple regression term with
$h(\vec{x})^\top = (1, x_1, \ldots, x_p)^\top $}. Concerning the nugget of the
surrogate, we perform a sigmoid transformation in order to sample all covariance
hyper-parameters with multivariate Gaussian distributions as discussed in
Section \ref{subsec:annealing_trans}. That is, we introduce an auxiliary
component $z_\delta$ and extend the vector of hyper-parameters $\boldsymbol\phi$
to $\mathbb{R}^{p+1}$. Finally, we compute the nugget as

\begin{align}
\theta_\delta = \frac{1-l_b}{1+\exp(-z_\delta)} + l_b ,
\end{align} 
where $l_b$ is the lower bound, which is set equal to $10^{-12}$ following the
discussions in \cite{Ranjan2011}. To incorporate the algorithm to the 
length-scale hyper-parameters, the sampling has been performed in logarithmic
space to avoid additional concerns for the non-negative restrictions
imposed to the aforementioned variables as in other sampling schedules
\citep{Neal1997}. The initial values of the algorithm, equation
\eqref{eq:meta_prior}, are set to a uniform distribution in a wide practical
range, that is the interval $\left[-7,7\right]$ for the length-scales. For the
nugget, a non-informative truncated beta distribution in the interval $\left[
l_b, 1 \right]$ has been considered.

The code was implemented in MATLAB and all examples were run in a GNU/Linux
machine with an Intel i5 processor with 8 Gb of RAM. For the purpose of
reproducibility, the code used to generate the examples in this paper is
available for download at \url{http://github.com/agarbuno/ta2s2_codes}.

\added{In order to contrast our proposed methodology with existing ones, we take
the Particle Learning sampler PLGP \citep{Gramacy2009} as a benchmark. This
sampler has proven effective for sampling the posterior distribution of the
hyper-parameters of a Gaussian process by means of tempering in a data-oriented
manner, \ie by feeding subsets of the training runs in each annealing level.
Although the proposed sampler can be implemented in an on-line fashion akin to
PLGP, we resort only at comparing them as strategies in batch applications. Note
that the extension to on-line learning tasks can be done by regarding the
posterior of a subset of data as the prior for the next set of training
runs. This can be followed easily as the re-weighting of the samples by a 
data-oriented alternative to \eqref{eq:unweight} can help adjust the importance
of the samples.}

\added{Following the discussion of \cite{Kohonen2006} proper scoring rules
should be used in order to compare the probability statements made by the
Gaussian process model resulting from the samples used to marginalise the
predictive posterior. In the context of Gaussian processes, both prediction and
error estimation are used to assess the quality of the surrogate, \ie the
estimated mean and variance. If a local scoring rule such as the negative
logarithm of predictive density (NLPD) is used to evaluate the generated
samples, we risk penalising heavily over-confident predictions and treat with
less rigour  under-confident far-off predictions. This is not desirable since it
is known that the full Bayesian treatment in Gaussian processes is preferred for
better error estimation in uncertainty analysis \citep{Kennedy2001}. In
contrast, by using distance-sensitive scoring rules such as the continuously
ranked probability score (CRPS) we aim for better placement of probability mass
near target values, although not exactly placed at the target. It is defined as}

\ifdraft{\begin{color}{blue}}{\begin{color}{black}}
\begin{align}
\text{CRPS}(F,x) = \int_{-\infty}^{\infty} (F(y) - \vec{1}\{y \geq x\})^2 dy,
\end{align}
\end{color}
\added{where $F$ is the cumulative predictive distribution and $x$ is the point
where it is verified. We assume a Gaussian approximation for the predictions
made by the Gaussian process emulator and use the mixture model expressed in
equations \eqref{eq:mix_mean} and \eqref{eq:mix_sigma} to be able to use the
complete mixture expression developed in \citep{Grimit2006} which we include for
completeness. That is, for a mixture of Gaussians the CRPS can be written as}

\ifdraft{\begin{color}{blue}}{\begin{color}{black}}
\begin{align}
\text{CRPS}\left( \sum_{m = 1}^N \omega_m \, \mathcal{N} ( \mu_m, s^2 )\, , \, x \right) = & 
\sum_{m = 1}^N \omega_m  \, A (x - \mu_m, s^2_m ) \, - \,  \nonumber \\
& \frac{1}{2} \sum_{m = 1}^N \sum_{n = 1}^N \omega_m \omega_n \, A\left(\mu_m - \mu_n, s^2_m + s^2_n\right), 
\end{align}
\end{color}
\added{\noindent where $\omega_i$ denotes the weight of the sample, $\mu_i$ is
the mean of $x$ given by sample $i$ and $s^2_i$ the corresponding estimated
variance. The function $A(\cdot, \cdot)$ is defined as }

\ifdraft{\begin{color}{blue}}{\begin{color}{black}}
\begin{align}
A(\mu,\sigma^2 ) = 2 \sigma f_{\mathcal{N}}\left( \frac{\mu}{\sigma}\right) +\mu \left( 2 \, F_{\mathcal{N}} \left( \frac{\mu}{\sigma}\right) -1 \right), 
\end{align}
\end{color}
\added{where $f_{\mathcal{N}}(\cdot)$ and $F_{\mathcal{N}}(\cdot)$ denote the
density and cumulative functions of a standard Gaussian random variable. It
should be noted that the CRPS does not possess an analytic expression for every
probability function used for prediction \citep{Grimit2006}, thus the choice of
using a mixture of Gaussians for the predictive posterior distribution instead
of $t$-distributions. }

\added{Other alternatives for scoring rules for Gaussian processes could be the
bootstrapped variance predictor of \cite{Hertog2006}. This predictor aims to
estimate the variance of the Gaussian process independent of the set of points
used for training. However, such approach leads to prefer under-confident
predictions not exactly around the target value and relies in the assumption of
infinite repeatability of the simulator experiments.This is not satisfied by
Bayesian analysis of computer code output (BACCO), since by assumption the
generation of training runs is limited by computational cost.}

\subsection{Franke's function} \label{subsec:franke}

Franke's function has been used to test Gaussian process emulators
\citep{Haaland2012}. Its complexity stems from the presence of two peaks and one
dip in its landscape. Let $f:\left[ 0, 1 \right]^2 \rightarrow \mathbb{R}$ be
such that

\begin{align}
f(\vec{x}) = \, & \, 0.75 \, \exp \left( - \frac{(9 x_1 - 2)^2}{4} - \frac{(9 x_2 - 2)^2}{4}
\right) + 0.75
\, \exp \left( - \frac{(9 x_1 + 1)^2}{49} - \frac{9 x_2 + 1}{10} \right) \nonumber\\
 \, & + 0.5 \, \exp \left( - \frac{(9 x_1 - 7)^2}{4} - \frac{(9 x_2 - 3)^2}{4} \right) - 0.2\, \exp
\left( -
(9x_1 - 4)^2 - (9x_2 -7)^2\right). 
\end{align}

To train the emulator, 20 design points were chosen using Latin hypercube
sampling (LHS). For testing purposes, 100 independent design points were chosen
by a second LHS. \replaced{Figure \ref{subfig:franke-level} shows the multi-
modal integrated log-posterior for a fully-parametrised Gaussian process
\citep{Garbuno2015}}{ Figure \ref{subfig:franke-level} shows that the log-
posterior distribution exhibits multi-modality}. Region A contains a mode with
no preference for any dimension. Regions B and C depict different asymptotic
behaviours of the emulator. In region B, the emulator behaves as linear
regression model, as noted by \cite{Andrianakis2012}. Region C corresponds to a
model which disregards the first dimension. \replaced{In Figure 
\ref{subfig:franke-sample} a set of samples obtained by applying TA$^2$S$^2$ is showed,
illustrating the ability to overcome possible multi-modal distributions that
arise in Bayesian analysis of expensive computer codes.}{ By applying
TA$^2$S$^2$, we obtain robust approximations to the MAP estimate, as shown in
Figure \ref{subfig:franke-sample}. This has been done by running the algorithm
with $N=2000$, obtaining 5 intermediate distributions. A MAP estimate of such
inference problem results in a Root-mean-squared error (RMSE) of 0.1557, in
contrast with the mixture estimator in equation \eqref{eq:mix_mean}, which
achieves a RMSE of 0.1069. This is due to the MAP estimate being located in
Region (B), which is a linear model of the inputs. The mixture model helps
alleviate this shortcoming by introducing emulators with local dependency among
the training runs.}

\begin{figure}[H]
\centering
\subfloat[Level curves]{\includegraphics[draft=false,width=.35\linewidth]{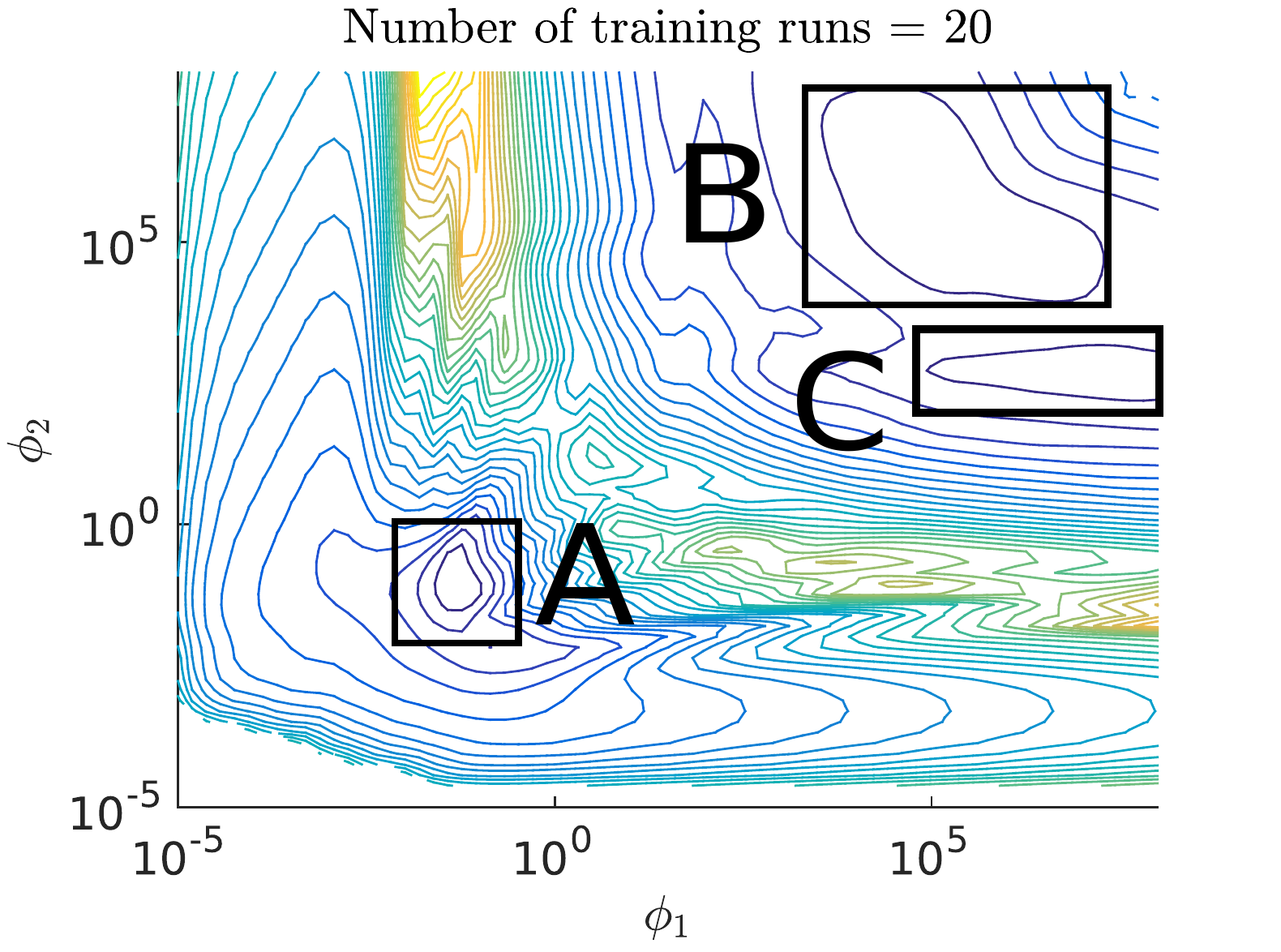}
\label{subfig:franke-level}}
\subfloat[TA$^2$S$^2$ samples]{\includegraphics[draft=false,width=.35\linewidth]{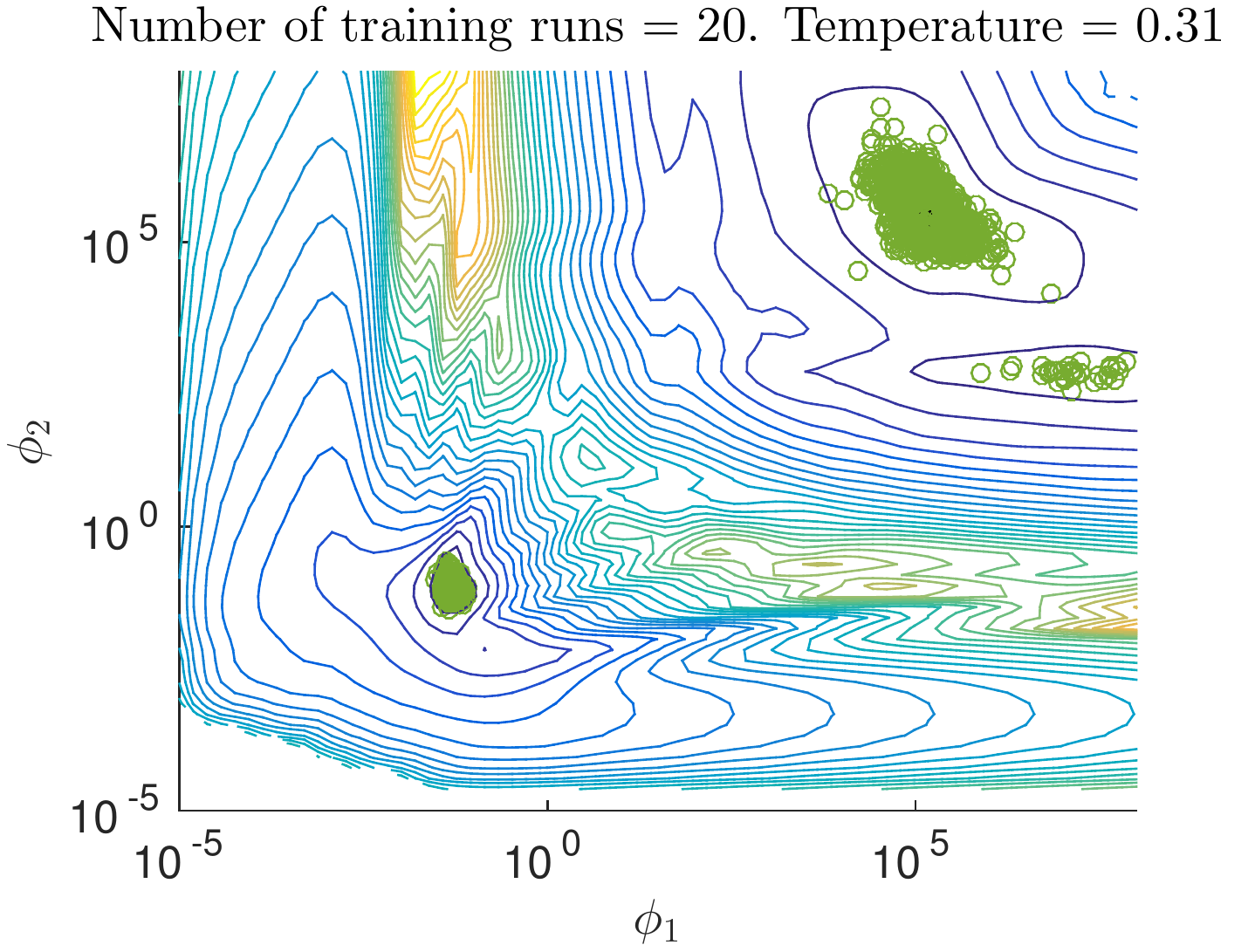}
\label{subfig:franke-sample}} 
%\subfloat[Residual plots]{\includegraphics[draft=false,width=.32\linewidth]{franke-residuals-eps-converted-to.pdf}
%\label{subfig:franke-residual}}
\caption{Projection of the negative log-posterior curves in the two dimensional
length-scale space for Franke's simulator using a fully-parametrised Gaussian
process. The minimum possible value of $10^{-12}$ for the nugget
$\boldsymbol\phi_\delta$ has been used for the projection and the temperature
has been let to reach a practical zero, thus retrieving samples from the modes.}
\end{figure}

\added{To contrast the proposed sampler against the PLGP benchmark, a set of 100
experiments were run. In each experiment, a sample of size 100 of length-scale
hyper-parameters was obtained by each method. This sample size was achieved by
thinning the TA$^2$S$^2$ results when a chain of length 2000 was constructed in
every annealing level. The quality of the probability statements made from both
results were compared by means of the CRPS, as depicted in Figure
\ref{fig:crps_franke}. Note that the training set was the same for each
experiment and variations in the results among experiments are mainly because of
the stochastic nature of the sampling schemes. In Figure \ref{subfig:franke-box1} 
the boxplots for the CRPS computed from the samples are shown. Figure
\ref{subfig:franke-box2} shows the same plots but with an exponential prior for
all hyper-parameters with rate 0.2, \ie $\lambda = 5$. This choice of a prior
distribution was made since the PLGP software assumes an exponential prior for
both the length-scales and nugget term \citep{Gramacy2009}. In both settings,
the proposed sampler outperforms both the PLGP alternative and the MAP estimate.
The latter was calculated from the samples generated by TA$^2$S$^2$. Our
experiments demonstrated that the MAP estimated this way usually corresponds to
the one found by local optimisation routines such as Nelder-Meade or BFGS. The
variation in the scores of the MAP illustrates the multi-modality properties of
the integrated posterior. All the MAP estimates reported in the remainder are
calculated based on this observation. Additionally, it can be seen that the PLGP
results seem to contain those achieved by the most probable candidate. In this
experiment either using a sample from PLGP or MAP translate in comparable
results.}

\begin{figure}[H]
\centering
\subfloat[Reference prior.]{\includegraphics[draft=false,width=.32\linewidth]{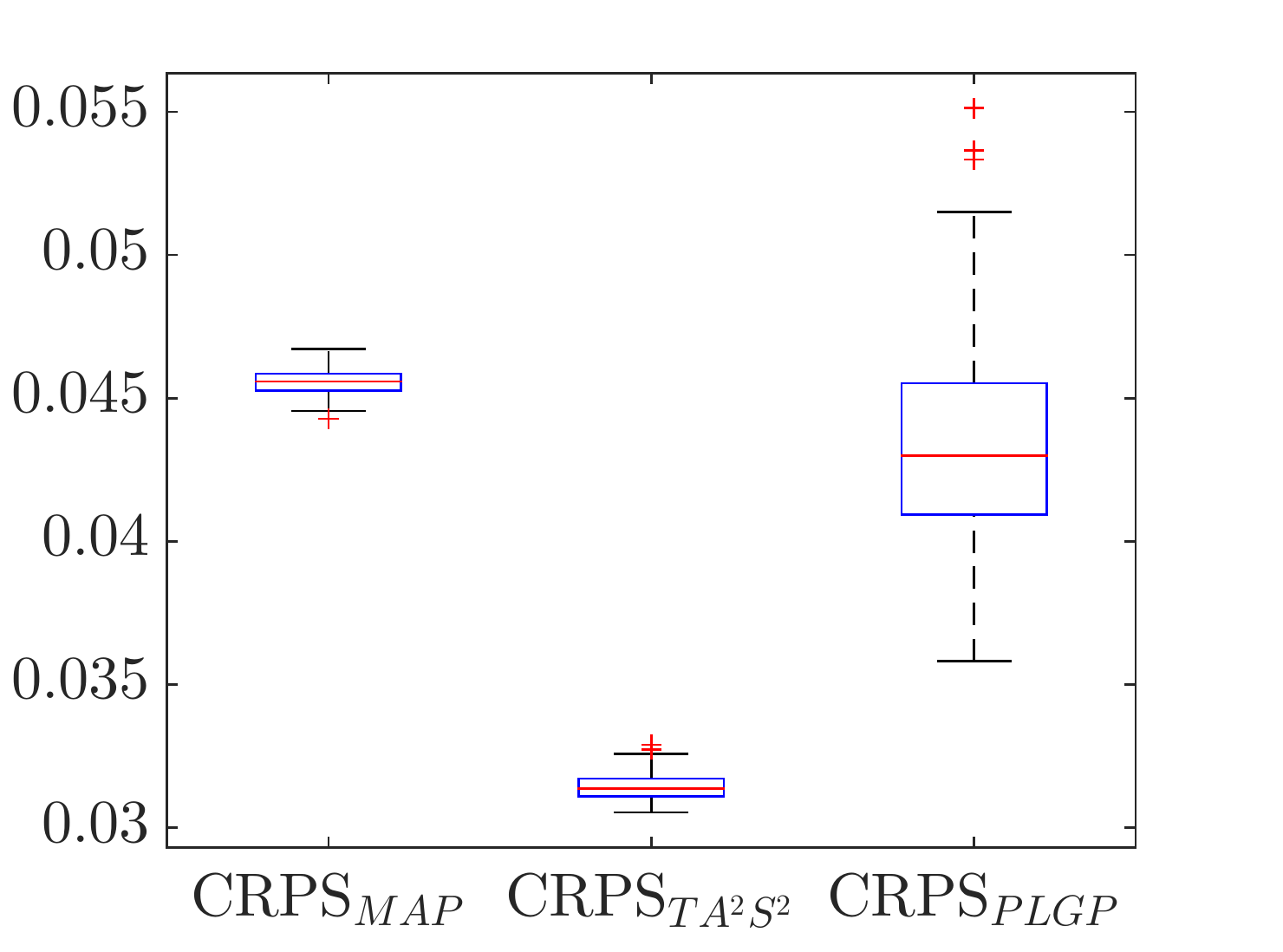}
\label{subfig:franke-box1}}
\subfloat[Exponential prior, $\lambda = 5.$]{\includegraphics[draft=false,width=.32\linewidth]{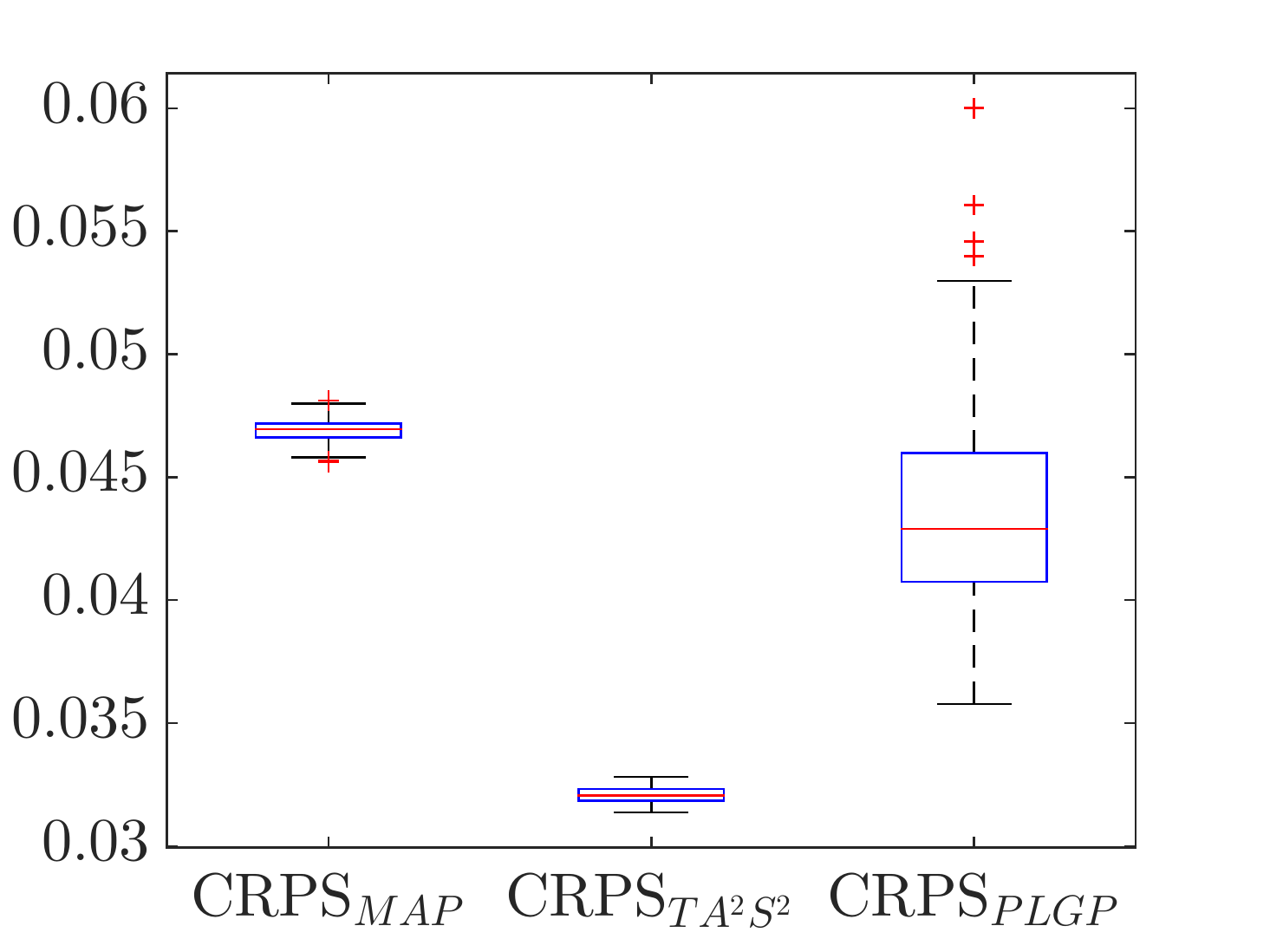}
\label{subfig:franke-box2}} 
\caption{Boxplots of CRPS comparing MAP, TA$^2$S$^2$ and PLGP. In both cases the proposed sampler outperforms PLGP.}
\label{fig:crps_franke}
\end{figure}

\subsection{Nilson-Kuusk model} 

This simulator models the reflectance of a homogeneous plant canopy. Its 
five-dimensional input space includes the solar zenith angle, the leaf area
index, the relative leaf size, the Markov clumping parameter and a model
parameter $\lambda$ \citep[see][for further details on the model itself and the
meaning of the inputs and output]{Nilson1989}. For the analysis presented in
this paper, a single output emulator is assumed and the set of the inputs have
been rescaled to fit the hyper-rectangle $[0,1]^5$ as in \cite{Bastos2009}.

In this experiment, both samplers were used to train a Gaussian process emulator
with a dataset of 100 simulation runs. The test set consisted of a different set
of 150 training runs. Both datasets, whose design points were generated through
LHS, were obtained from the GEM-SA software web page
(\url{http://ctcd.group.shef.ac.uk/gem.html}). On average, a total of 10
tempered distributions were used in the annealing schedule, while keeping the
sampling as $N=5000$ in each level. A thinned sample of 100 experiments was
recovered by the end of each TA$^2$S$^2$ run to compare results.\deleted{The
RMSE for both the MAP and mixture is 0.0214. This seems to indicate that given
the training set, the posterior distribution of the length- scales is highly
concentrated around one mode. The boxplots in figure \ref{subfig:kuusk-results}
shows the distribution of the length-scales in the five-dimensional log-space.
There is strong evidence that the fifth dimension is more important and there is
clearly and ordered triplet $(\phi_5, \phi_1, \phi_2)$ for the sensitivity of
the output to such variables. However, there is not enough information to make a
more informed judgement regarding $\phi_3$ and $\phi_4$. The standardised
residuals confirm the assumption of little difference when considering the
sensitivity of the simulator to the dimensions associated with $\phi_3$ or
$\phi_4$.}

\added{The results shown in Figure \ref{fig:kuusk-res} demonstrate again the
overall improved performance of using the proposed sampler in contrast with the
benchmark. In this case, one set of experiments (50 iterations) consisted on
making inference with the reference prior discussed previously, while the second
set (50 iterations) used a common exponential prior. Both samplers outperform
the MAP estimate which provides evidence that a more complete uncertainty
analysis can be carried out if instead one turns to a full Bayesian inference
scheme. Additionally, it is worth noting that Figures \ref{subfig:kuusk-box1}
and  \ref{subfig:kuusk-box2} show that the MAP changes with the prior used. This
is not a matter of concern, as it is consistent with the notion that small
amount of data is being used for inference and the probabilistic model of the
observables is not dominating the prior beliefs. }

\begin{figure}[H]
\centering
\subfloat[Reference prior.]{\includegraphics[draft=false,width=.32\linewidth]{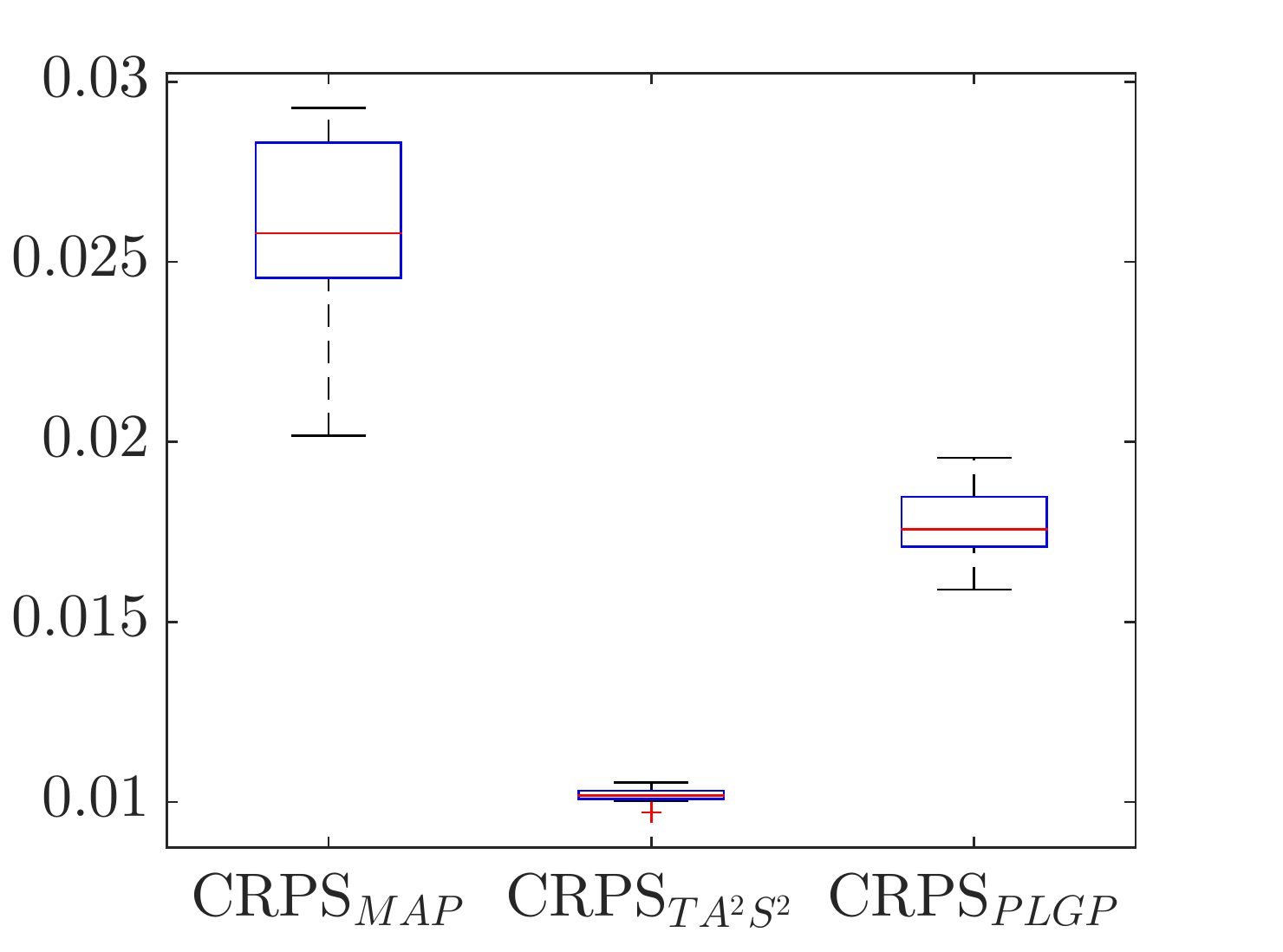}
\label{subfig:kuusk-box1}}
\subfloat[Exponential prior, $\lambda = 5$.]{\includegraphics[draft=false,width=.32\linewidth]{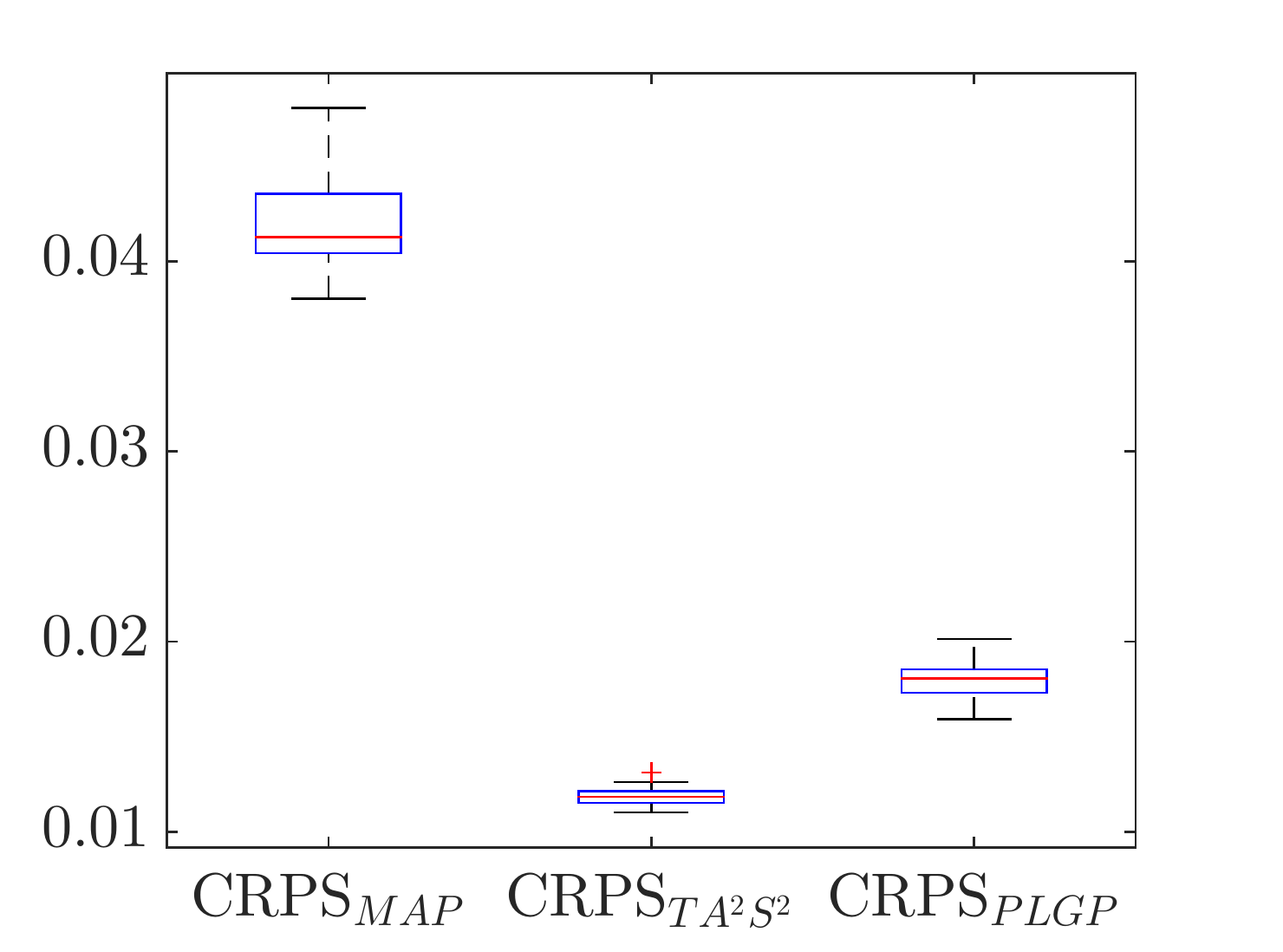}
\label{subfig:kuusk-box2}} 
\subfloat[Sample from TA$^2$S$^2$ for each length-scale.]{\includegraphics[draft=false,width=.32\linewidth]{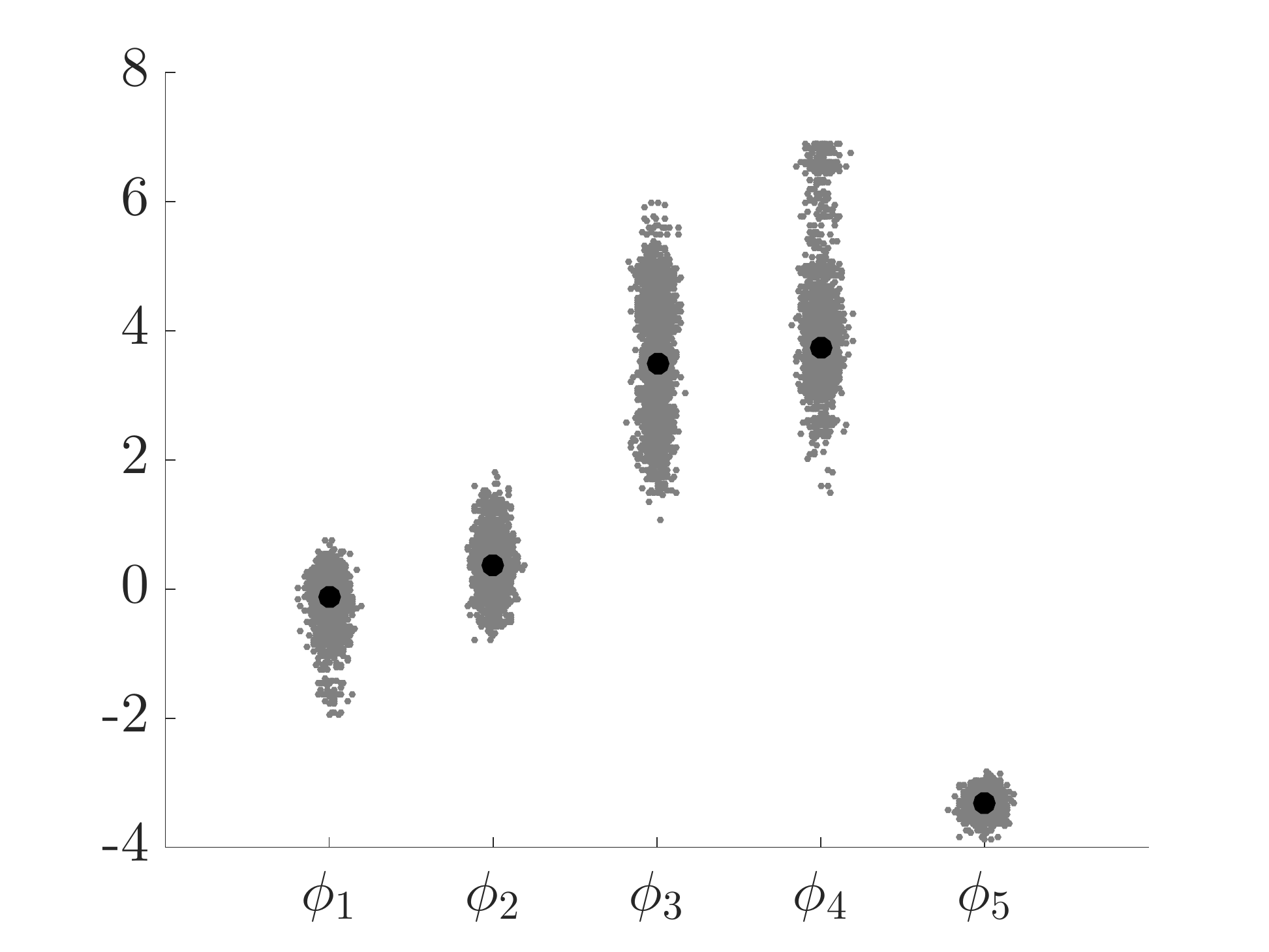}
\label{subfig:kuusk-scatter}}
\caption{Nilson-Kuusk results obtained after 50 experiments were run. The scatter plot \eqref{subfig:kuusk-scatter} is drawn from a single experiment, in black the MAP estimates.}
\label{fig:kuusk-res}
\end{figure}

\subsection{Wing weight model}

This simulator of the weight of the wing of a light aircraft
\citep{Forrester2008a} has been used for input screening. Coupled with a Gaussian
process emulator with the squared exponential kernel in equation 
\eqref{eq:covariance-func}, a sensitivity analysis of the wing weight with
respect to each input variable can be performed. The model is given by

\begin{align}
f(\vec{x}) = 0.036 \, S_w^{0.758} \, W_{fw}^{0.0035} \, \left( \frac{A}{\cos^2(\Lambda)} \right)^{0.6} \,
q^{0.006} \, \lambda^{0.04} \, \left( \frac{100\,t_c}{\cos(\Lambda)} \right)^{-0.3} (N_z\, W_{dg})^{0.49} \,
+ \, S_w \, W_p, 
\end{align}

\noindent where the input variables and the range of their values are summarised
in Table \ref{tab:wing}.\deleted{In order to focus on determining the most
relevant inputs on the correlation structure of the Gaussian process, it is
assumed that $h(\vec{x}) = 1$.} For this problem, the evaluation of the
reference prior is prohibitive since it scales with the number of dimensions
\citep{Paulo2005}. Thus, a uniform prior in the hyper-parameters'
log-space has instead been used for this experiment.

\begin{table}[H]
\centering
\begin{tabular}{|c|c|c|}
\hline
Input & Range & Description \\
\hline
$S_w $ & $ [150, 200]$ & Wing area \\
$W_{fw} $ & $ [220, 300]$ &	Weight of fuel in the wing \\
$A $ & $ [6, 10]$& 	Aspect ratio \\
$\Lambda$ & $ [-10, 10]$& 	Quarter-chord sweep \\
$q $ & $ [16, 45]$& 	Dynamic pressure at cruise \\
$\lambda $ & $ [0.5, 1] $&	Taper ratio \\
$t_c $ & $ [0.08, 0.18]$& 	Aerofoil thickness to chord ratio \\
$N_z $ & $ [2.5, 6]$& 	Ultimate load factor\\
$W_{dg} $ & $ [1700, 2500]$& 	Flight design gross weight \\
$W_p $ & $ [0.025, 0.08] $&	Paint weight \\ 
\hline
\end{tabular}
\caption{Inputs of the wing weight model.}
\label{tab:wing}
\end{table}

\noindent The inputs were rescaled to the 10-dimensional unit hypercube
$[0,1]^{10}$. Two LHS samples of size 100 and 300 were chosen as training and
testing sets respectively. At each annealing level, 5000 samples were generated,
achieving convergence after 15 levels on average. As before, a thinned sample of
100 was kept to compare results with the benchmark in each experiment. A total
of 50 experiments we run in this case.

In Figures \ref{subfig:wings-box1} and \ref{subfig:wings-box2}, it can be noted
that by using sampling one can obtain an improved version of the probabilistic
statements made by the surrogate. Figure \ref{subfig:wings-box1} shows that the
proposed sampler outperforms the benchmark, concentrating its samples
around the mode of the posterior distribution. This is shown by the location and
spread of the CRPS for both TA$^2$S$^2$ and MAP results, which are similar.
Although the CRPS with TA$^2$S$^2$ exhibits better performance as can been seen
from the location and spread of the boxplot. PLGP in the case of Uniform priors
seem to be sampling from other areas with less spread than that of the proposed
algorithm. However, if an exponential prior distribution is used for both
length-scales and nugget, Figure \ref{subfig:wings-box2}, TA$^2$S$^2$ is
outperformed by the benchmark. Note how the location of the mode changes
dramatically, showed by both the means of the boxplot of the CRPS of the MAP,
and by the difference of the samples plotted in Figures \ref{subfig:wings-box3}
and \ref{subfig:wings-box4}. Nonetheless, the proposed sampler is capable of
offering improved probabilistic statements for the regression task compared to
the MAP estimate.

\deleted{The RMSE for the MAP estimate was 3.1719, whilst 3.1918 for the
mixture. This could suggest that there is little value in computing the mixture.
However, the simple analysis presented in figure \ref{subfig:box-wings} not only
shows that there is a subset of influential variables on the output of the
model. It also reveals that the less influential of the inputs (2, 4, 5) are not
concentrated around a single mode and thus the ranking of their influence on the
output is not unique. This information would not be available had the mixture of
Gaussian processes not been employed. The standardised residuals shown in Figure
\ref{subfig:res-wings} are evidence that the posterior distribution is highly
concentrated in the modes of the relevant inputs, and its application for error
prediction is not affected by the variation of the remaining inputs, though the
mixture model consistently narrows the spread of the residuals. By fully
characterising the optimal approximations with the TA$^2$S$^2$ algorithm, the
structural uncertainty of the model is taken fully into account.}

\begin{figure}[H]
\centering
\subfloat[Uniform prior.]{\includegraphics[draft=false,width=.32\linewidth]{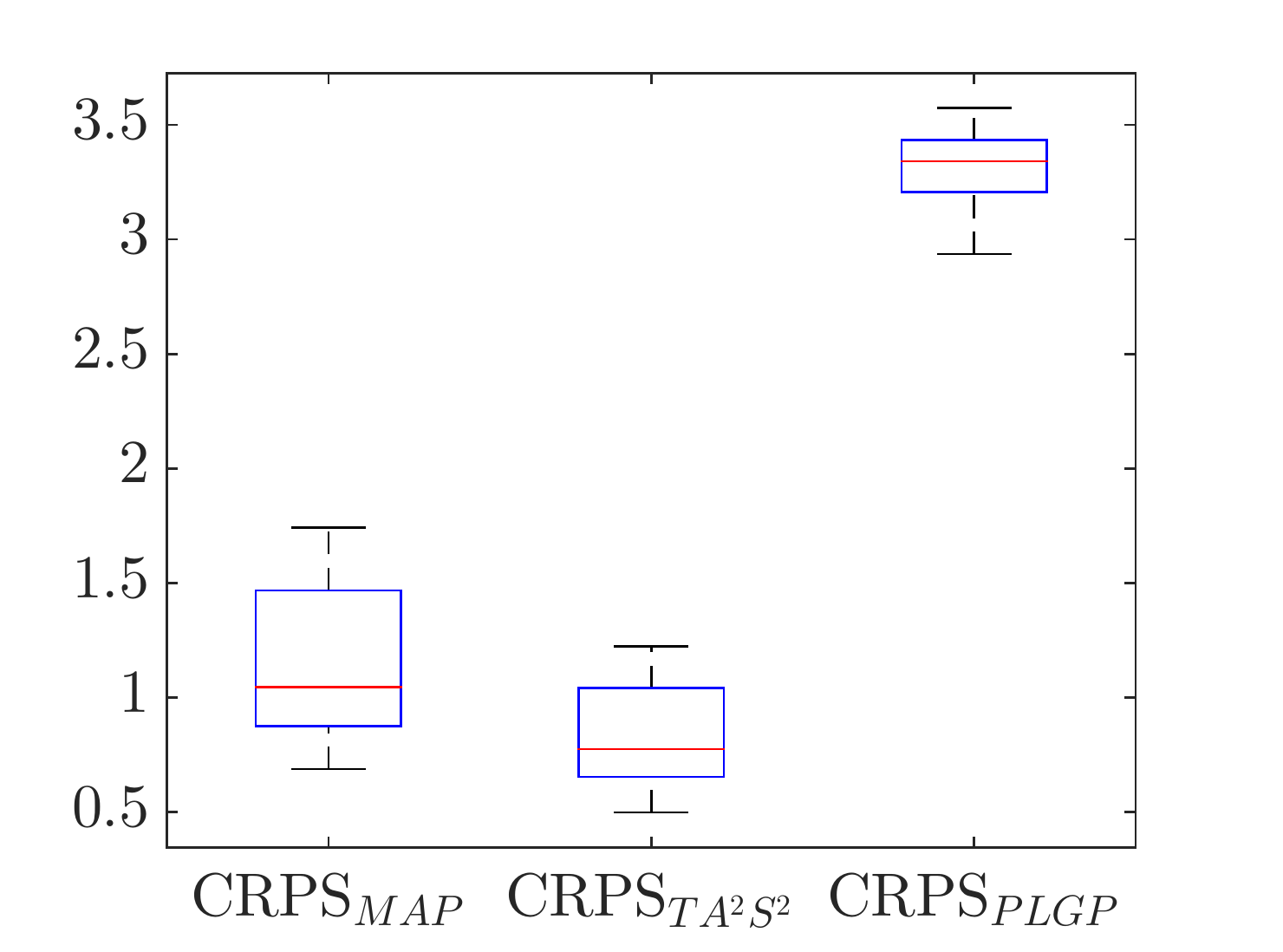}
\label{subfig:wings-box1}} \qquad
\subfloat[Exponential prior, $\lambda = 5$.]{\includegraphics[draft=false,width=.32\linewidth]{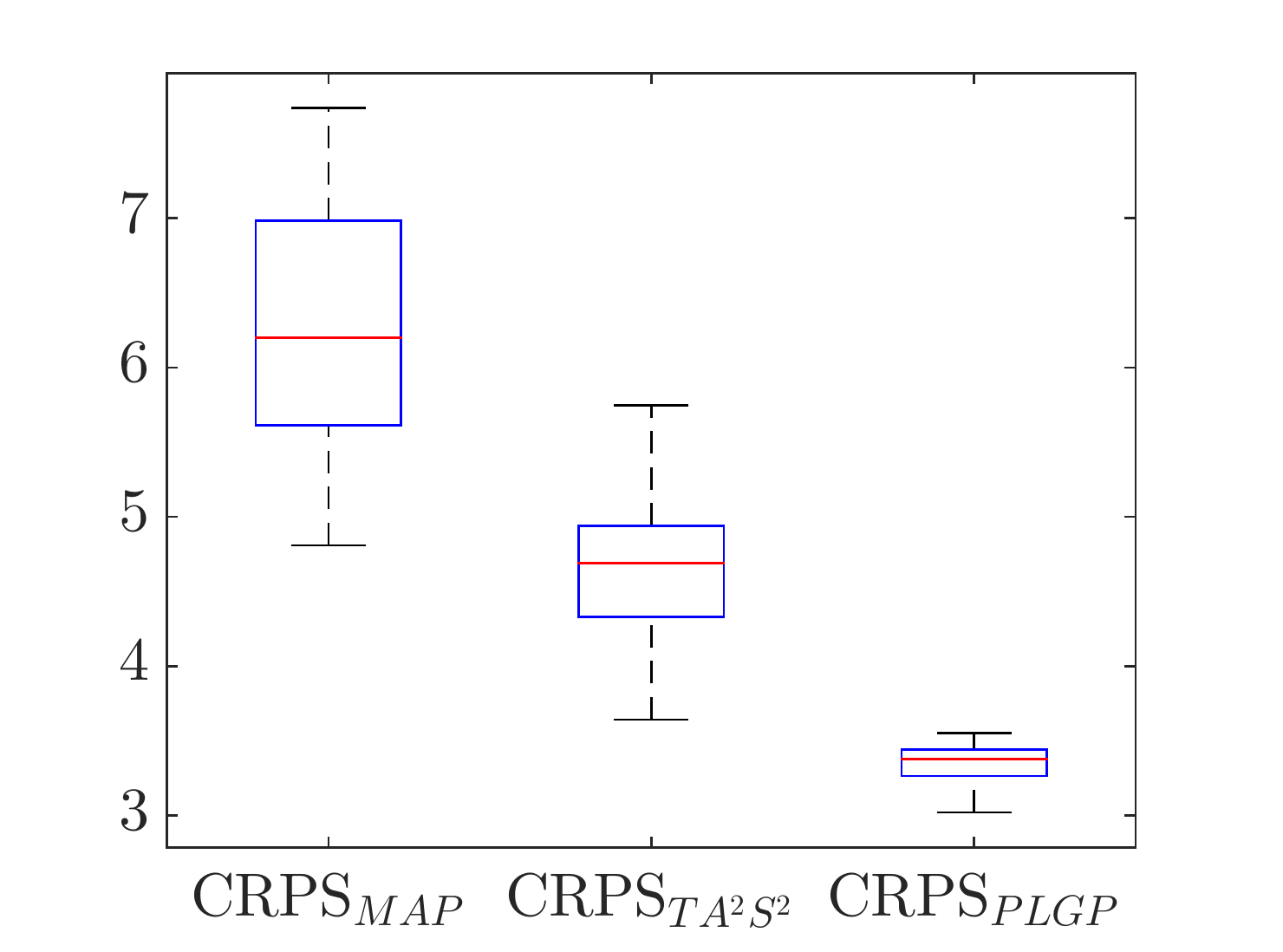}
\label{subfig:wings-box2}} \\
\subfloat[Sample from TA$^2$S$^2$ for each dimension. Uniform prior.]{\includegraphics[draft=false,width=.32\linewidth]{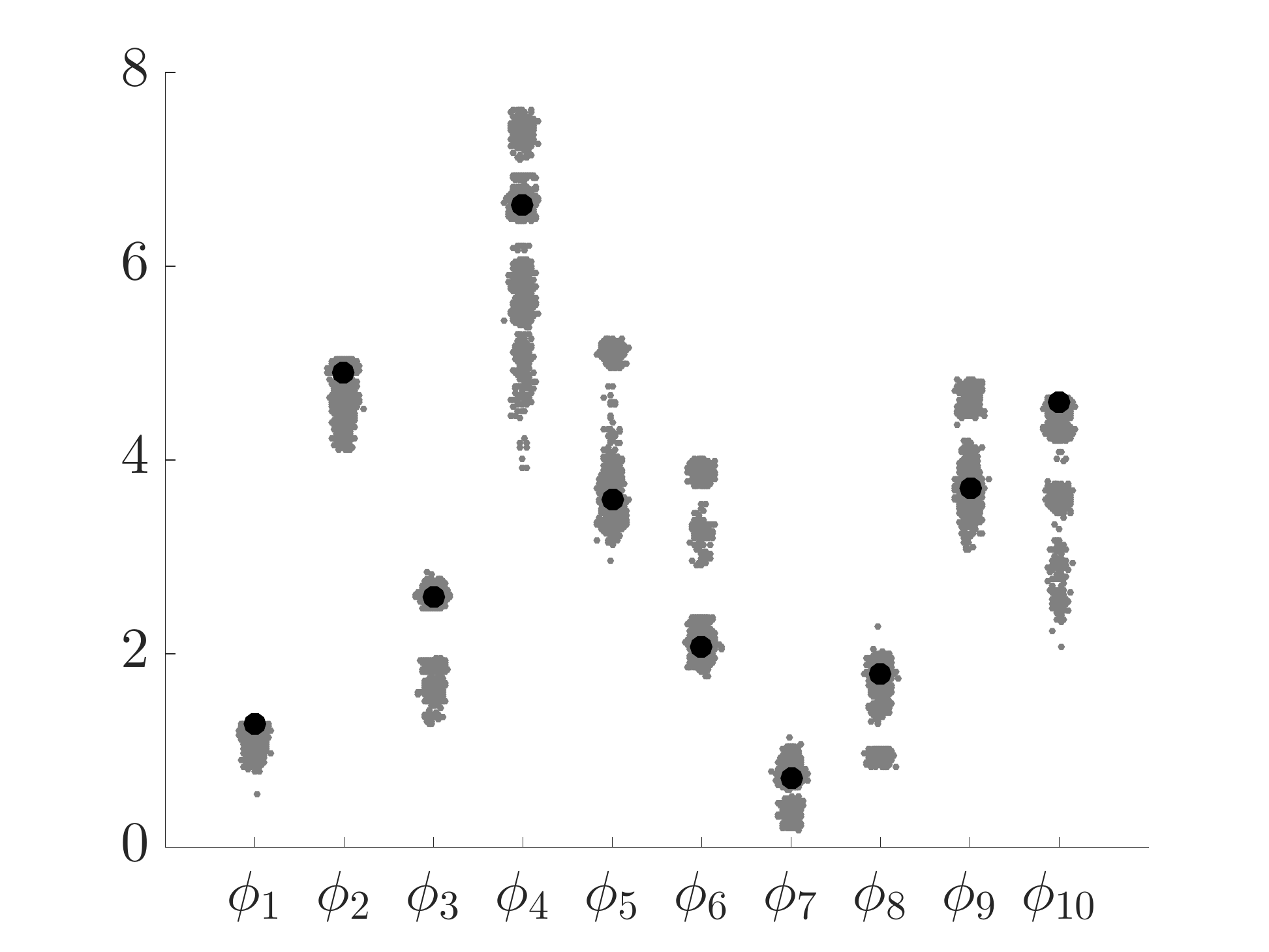}
\label{subfig:wings-box3}} \qquad
\subfloat[Sample from TA$^2$S$^2$ for each dimension. Exponential prior.]{\includegraphics[draft=false,width=.32\linewidth]{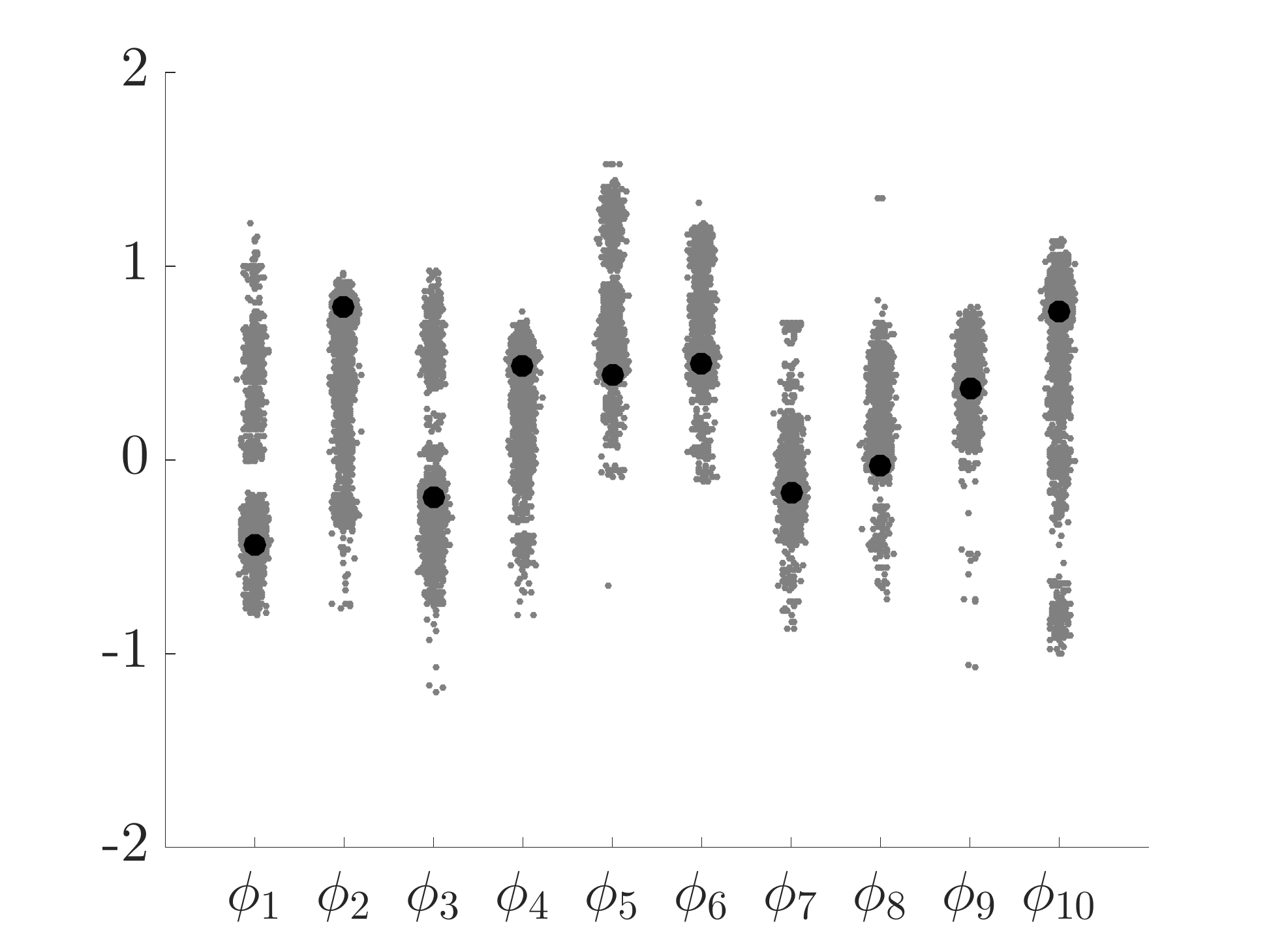}
\label{subfig:wings-box4}}
\caption{Results for the wing-weight model using the TA$^2$S$^2$ and PLGP algorithms. The scatter plot \eqref{subfig:kuusk-scatter} is drawn from a single experiment, in black the MAP estimates.}
\label{fig:wings-res}
\end{figure}

% ----------------------------------------------------------------------------
\section{Conclusions} \label{sec:conc}

This paper proposes a method, Transitional Annealed Adaptive Slice Sampling
(TA$^2$S$^2$), to sample from the posterior distribution of the Gaussian
process' hyper-parameters. This is known to be a problem where multi-modal
distributions are encountered. TA$^2$S$^2$ combines Slice Sampling for delayed-
rejection, Transitional Markov Chain Monte Carlo (TMCMC) for efficient
parallelisation and Markov chain growth, and Asymptotically Independent Markov
Sampling (AIMS) for driving the annealing schedule. The delayed-rejection
feature of Slice Sampling provides improved mixing which is desirable in highly
correlated spaces. Additionally, the proposed algorithm provides a sampling
scheme with no burn-in periods for Markov the chain, since each sample generated
follows the desired distribution. \added{This is advantageous in Gaussian
process applications where the computational burden of the sampling scheme is
dominated by the inversion of the correlation matrix. Moreover, }efficient
coverage of the sampling space is achieved by annealing and an extension of
Slice Sampling to Sequential Monte Carlo.

The examples presented show how the method is capable of efficiently exploring
multi-modal distributions providing a better alternative to MAP estimates or
traditional MCMC methods for Gaussian processes applications.
\added{Furthermore, TA$^2$S$^2$ allows the generation of samples more
efficiently than in traditional MCMC applications, where longer chains are
needed to ensure good coverage of the sampling space. This is particularly
relevant in the context of computer experiments and Engineering, where it is
usually considered that the cost of marginalising the hyper-parameters does not
justify the additional computational burden. The proposed method justifies the
cost of sampling the hyper-parameters by providing robust estimates of the
predicted error. This is reflected through the continuously ranked probability
score (CRPS) which shows that the marginalised emulator outperforms MAP
estimates. Moreover, the proposed sampling scheme performs as good as the PLGP
benchmark, which justifies the annealing by means of functional tempering rather
than by subsets of data, which is also sensitive to data permutations.}

\replaced{The proposed algorithm could also be employed for global optimisation
problems, by allowing the temperature reach a practical zero. This has
application in other areas of Bayesian inference problems and machine learning.
As discussed previously, the algorithm can be used in active learning schedules
by means of Bayes' Theorem.}{ Additionally, it has also been shown that
TA$^2$S$^2$ can be used in other Bayesian inference contexts (such as
variational inference) and for other surrogates (such as a Bayesian neural
networks)}

The computational cost of the proposed sampler is dominated by the inversion of
the covariance matrix in the integrated posterior distribution. This limitation
is inherent to any MCMC schedule used for Gaussian processes. Further
strategies to improve the speed of the sampler could be developed, such as the exact evaluation of the log-posterior for candidates where the probability of
acceptance is high. For such strategy, a first order Taylor expansion of the
objective function could be a feasible enhancement. This is material for future
research.

% ----------------------------------------------------------------------------

\section*{Acknowledgements}

The first author gratefully acknowledges the Consejo Nacional de Ciencia y
Tecnolog\'{\i}a (CONACyT) for the award of a scholarship from the Mexican
government for graduate studies. \added{The authors would like to thank
additionally an anonymous reviewer for his/her comments to improve the quality
of this paper.}

% ----------------------------------------------------------------------------
%	REFERENCE LIST
%-----------------------------------------------------------------------------

\end{document}